\documentclass[12pt]{article}
\usepackage{cite}
\usepackage{amsmath,amssymb,amsfonts,mathtools}
\usepackage{graphicx}
\usepackage{textcomp}
\usepackage{xcolor}
\usepackage{colortbl}
\usepackage{amsthm}
\usepackage[hidelinks]{hyperref}
\usepackage{caption}
\usepackage{array}
\usepackage{mathrsfs}
\usepackage{lmodern}
\usepackage{listings}
\usepackage{longtable}
\usepackage{booktabs}
\usepackage{enumitem}

\newcommand{\leanlistingfont}{\ttfamily}
\lstdefinelanguage{Lean4}{
  keywords={theorem, lemma, def, by, have, rw, simp, ring, native_decide, norm_num, omega, let, show, from, fun, exact, refine, intro, rintro, revert, constructor, ext, import, open, section, end, where, match, with, if, then, else, do, return, sorry, variable, instance, class, structure, namespace, noncomputable, deriving},
  morecomment=[l]{--},
  morecomment=[n]{/-}{-/},
  morestring=[b]",
  sensitive=true,
}
\usepackage{xpatch}
\makeatletter
\AtBeginDocument{\xpatchcmd{\@thm}{\thm@headpunct{.}}{\thm@headpunct{}}{}{}}
\makeatother

\newtheorem{defn}{Definition}[section]
\newtheorem{theorem}{Theorem}[section]

\newtheorem{prop}[theorem]{Proposition}

\newtheorem{example}{Example}[section]

\newtheorem{remark}[theorem]{Remark}

\newcommand{\F}{\mathbb{F}}

\newcommand{\rank}{\mathrm{rank}}

\newcommand{\wt}{\mathrm{wt}}

\newcommand{\et}{\mathrm{\acute{e}t}}

\begin{document}

\title{Formalizing building-up constructions of self-dual codes through isotropic lines in Lean}

\author{
Jae-Hyun Baek \\
Department of Artificial Intelligence and \\
Institute for Mathematical and Data Sciences \\
Sogang University, Seoul, Korea
\and
Jon-Lark Kim\\
Department of Mathematics and\\
Institute for Mathematical and Data Sciences \\
Sogang University, Seoul, Korea
}

\date{September 10, 2026}

\maketitle

\begin{abstract}
The purpose of this paper is two-fold. First, we show that, after a specified
form isometry, the two-coordinate reduction in the binary Hilbert-symbol
realization of Chinburg and Zhang is inverse to Kim's building-up construction,
up to permutation equivalence. Second,
for $q\equiv1\pmod4$, we develop a $q$-ary analogue of this
reduction-and-extension mechanism. The identity $c^2=-1$ yields the
isotropic line governing the split construction.  For every fixed ordered
pairing of the coordinates, we obtain a universal rank-$r$ boxed normal
form, where $r$ is the dimension of the intersection with the product of
these isotropic lines.
Applications include optimal self-dual $[6,3,4]$ and $[8,4,4]$ codes over
$\mathbb F_{5}$, optimal self-dual $[8,4,5]$ and $[10,5,6]$ codes over
$\mathbb F_{13}$, and a self-dual $[12,6,6]$ code over $\mathbb F_{13}$.  We also give
an exact repeated boxed realization of self-dual $[18,9,8]$ and
$[20,10,10]$ codes over $\mathbb F_{13}$, in which the split-boxed parent and its
building-up child occur in one complete generator matrix.
The algebraic core is formalized in Lean~4.
\end{abstract}

% \tableofcontents
% \newpage

%% ===========================================================================
\section{Introduction}
\label{sec:intro}
%% ===========================================================================

Coding theory is the study of reliable communication. Its modern foundations
were laid by Shannon's seminal paper~\cite{Sha}. Since then, many classes of
linear codes have been studied. In particular, the extended binary Hamming
code of length 8, the binary Golay code of length 24, and the ternary Golay
code of length 12 are self-dual, that is,
$\mathcal C=\mathcal C^{\perp}$.
Thus the classification or construction of self-dual codes has been one of the most active research problems~\cite{self-dual}.
Self-dual codes have wide connections with several mathematical areas including $t$-design theory~\cite{BacGab}, group theory~\cite{ConSlo}, number theory~\cite{BaDou},~\cite{ConSlo},~\cite{ShiChoShaSol}. Furthermore, they have applications in quantum error-correcting codes~\cite{CRSS}.

Kim~\cite{Kim2001} introduced ``building-up construction'' which proved that all binary self-dual codes can be constructed in this way. This construction starts from a binary self-dual code of length $2n$ and then adds a vector of length $2n$ of odd weight to construct a binary self-dual code of length $2n+2$. Kim--Lee~\cite{KimLee2009},~\cite{KimLee} generalized this method to construct self-dual codes over finite fields $\mathbb F_q$, where $q$ is even or $q \equiv 1 \pmod{4}$.
 In practice, this is a principal method for constructing examples from a
 split bilinear space. However, the formula for the new generator rows is
 often presented in terms of an auxiliary vector whose inner product with
 itself is $-1$. One of the central points of the present paper is that the
 corresponding correction coordinates are naturally organized by an
 isotropic line in a hyperbolic plane, and that the forward construction
 admits a precise row-wise and family-wise characterization in those terms.

 On the other hand, Kreck and Puppe~\cite{KreckPuppe} constructed binary
self-dual codes from the cohomology of three-manifolds. Chinburg and
Zhang~\cite{ChinburgZhang} developed the arithmetic analogue for rings of
$S$-integers. For a suitable set $S$ of places of a number field, the image
of the global-to-local map in first \'{e}tale cohomology is Lagrangian for
the sum of the local Hilbert pairings. Their Theorem~1.5 shows that, for
$K=\mathbb Q$, every binary self-dual code of length at least four occurs up
to permutation equivalence. We use precisely this realization theorem and
its boxed reduction; no further \'{e}tale-cohomological machinery is needed
below.

In this paper, we focus on $q \equiv 1 \pmod 4$. Over $\F_q$, self-dual codes are maximal totally isotropic subspaces for the standard Euclidean bilinear form on $\F_q^{2n}$. In this split case the elementary identity $c^2=-1$ is the basic algebraic input behind both the classical building-up construction and the structural description of the ambient quadratic space: it produces an isotropic line in $\F_q^2$, identifies the Euclidean plane with a hyperbolic plane, and makes the norm form $x^2+y^2$ split.

%The purpose of this paper is to make that common mechanism explicit and to push it simultaneously in three directions: constructive building-up, cohomological interpretation, and machine-checked verification.

We establish that each generator row of the building-up admits the isotropic line decomposition $r_i = -y_i \cdot (1, c, 0, \ldots, 0) + (0, 0, g_i)$, where $(1, c)$ spans an isotropic line in the hyperbolic plane. This decomposition provides a precise geometric interpretation of the correction terms in the building-up construction. We also give a constructive norm-form: once $c^2=-1$ and $2\neq 0$, every $a\in\F_q$ can be written explicitly as $a=x^2+y^2$, which in turn yields extension vectors of prescribed norm.

The odd-characteristic and binary inputs are distinct. For
\(q\equiv1\pmod4\), the element \(-1\) is a square in \(\F_q^\times\), so
the Euclidean plane is hyperbolic: the vectors
\[
h_1=(1,c),\qquad h_2=\frac12(1,-c)
\]
satisfy \(h_1\cdot h_1=h_2\cdot h_2=0\) and
\(h_1\cdot h_2=1\). In the binary theorem of
Chinburg--Zhang~\cite{ChinburgZhang}, the arithmetic image is Lagrangian for
a cohomological pairing; selected orthonormal coordinates identify that
pairing with the Euclidean dot product. The coordinate plane deleted in the
resulting binary reduction has Gram matrix \(I_2\) and is non-alternating.

In the present paper, this distinction organizes the binary comparison, the
split \(q\)-ary construction, its universal normal form, and selected
applications over \(\F_{5}\) and \(\F_{13}\).

Our main contributions are as follows.
\begin{enumerate}[label=(\arabic*)]

%\item We reinterpret the building-up rows through a split isotropic line in a hyperbolic plane and obtain both a one-step split building-up theorem and a split boxed-form theorem, together with row-wise and family-wise characterization and uniqueness over a fixed parent (Section~\ref{sec:building-up}).

%\item We develop the reverse side in the same split language: after fixing a distinguished row on the \(c\)-isotropic line and a pivot-adapted basis for the remaining rows, self-duality forces exactly the boxed relations appearing in the forward theorem. In this sense the split boxed form gives both a constructive extension theorem and a conditional normalization theorem (Section~\ref{sec:building-up}).

\item In the binary case, we distinguish the cohomological pairing from the
Euclidean dot product by a specified form isometry. In the resulting
Euclidean coordinates, boxed reduction deletes a two-coordinate block with
Gram matrix \(I_2\), the tail rows form a self-dual predecessor, and the
remaining coefficients satisfy Kim's correction formula
(Theorem~\ref{thm:binary-cz-kim}).

\item In the \(q\)-ary case, we prove the adapted reduction and extension
theorem, a forward split boxed construction, and its reverse statement under
an explicit adapted-presentation hypothesis
(Theorems~\ref{thm:qary-building-up}--\ref{thm:conditional-split-boxed-normalization}).

\item For each fixed ordered coordinate pairing, every Euclidean self-dual
code over \(\F_q\), \(q\equiv1\pmod4\), has the universal rank-\(r\) boxed
normal form of Theorem~\ref{thm:qary-universal-rank-r}, with
\(r=\dim(\mathcal C\cap U_c)\).

\item The applications retain the explicit small-length results over
\(\F_{5}\) and \(\F_{13}\) from the submitted manuscript and add one complete
repeated boxed realization of self-dual \([18,9,8]\) and \([20,10,10]\)
codes over \(\F_{13}\), displayed as an exact parent--child generator
(Section~\ref{sec:applications}).

\item The modular Lean~4 development formalizes the paper-facing algebraic
interfaces without \texttt{sorry} or \texttt{native\_decide}; independent
Challenge/Solution comparisons check the exact theorem statements.
\end{enumerate}

The cited arithmetic and coding-theoretic realization theorems remain
mathematical inputs. Lean verifies the finite-dimensional algebra and the
displayed construction, reduction, and normal-form interfaces; the distance
enumerations are independent exact computations.

The paper is organized as follows. Section~\ref{sec:background} fixes
notation. Section~\ref{sec:building-up} proceeds from the binary case to the
split \(q\)-ary case and finally to the universal normal form.
Section~\ref{sec:applications} contains the selected applications, and
Section~\ref{sec:conclusion} concludes.

%% ===========================================================================
\section{Preliminaries}
\label{sec:background}
%% ===========================================================================

%\subsection{Self-Dual Codes over Finite Fields}

We refer to~\cite{HuffmanPless} for basic concepts and facts and to~\cite{self-dual} for self-dual codes.

Throughout the paper, \(\F_q\) denotes a finite field. We equip \(\F_q^n\) with the Euclidean bilinear form
\[
\langle u,v\rangle := \sum_{i=1}^n u_i v_i.
\]
Here ``Euclidean'' refers to the coding-theoretic dot product, not to
positive definiteness. For \(u,v\in\F_q^n\), their distance and the weight
of \(u\) are, respectively,
\[
d(u,v):=\bigl|\{i\in\{1,\ldots,n\}: u_i\ne v_i\}\bigr|,
\qquad
\wt(u):=\bigl|\{i\in\{1,\ldots,n\}: u_i\ne0\}\bigr|=d(u,0).
\]
For a code \(\mathcal C\subseteq\F_q^n\) with at least two codewords, its
minimum distance is
\[
d(\mathcal C):=\min\{d(u,v):u,v\in\mathcal C,\ u\ne v\},
\]
and, when \(\mathcal C\ne\{0\}\), its minimum weight is
\[
\wt(\mathcal C):=\min\{\wt(u):u\in\mathcal C,\ u\ne0\}.
\]
Since \(d(u,v)=\wt(u-v)\), the minimum distance of a nonzero linear code
equals the minimum weight of its nonzero codewords:
\(d(\mathcal C)=\wt(\mathcal C)\). For a linear code
$\mathcal C\subseteq\F_q^n$, we write
\[
\mathcal C^\perp :=
\{v \in \F_q^n : \langle v,c\rangle =0\text{ for all }c\in\mathcal C\}.
\]
For subspaces \(\mathcal A,\mathcal B\subseteq V\), put
\(\mathcal A\cap\mathcal B:=\{v\in V:v\in\mathcal A\text{ and }
v\in\mathcal B\}\).  For a matrix \(G\) over \(K\),
\(\langle G\rangle_K\) denotes the code spanned by its rows.

\begin{defn}[self-dual code]\label{def:sd}
{\rm
A linear code $\mathcal C\subseteq\F_q^n$ is \emph{self-dual} if
$\mathcal C=\mathcal C^\perp$~\cite{Pless1965}.
}
\end{defn}

\begin{defn}[Lagrangian subspace]\label{def:lagrangian}
{\rm
Let \(V\) be a finite-dimensional vector space over \(\F_q\), equipped
with a nondegenerate symmetric bilinear form. A subspace \(L\subseteq V\)
is called \emph{Lagrangian}~\cite{Mon,Ser} if \(L=L^\perp\).
}
\end{defn}

\noindent
For a nondegenerate form,
\(\dim L^\perp=\dim V-\dim L\). Hence \(L=L^\perp\) if and only if
\(L\) is totally isotropic and \(2\dim L=\dim V\). In particular, a
bilinear space admitting a Lagrangian subspace is even-dimensional.
Applied to the standard Euclidean form, this says that
\(\mathcal C\subseteq\F_q^n\) is self-dual if and only if it is totally isotropic
and \(2\dim\mathcal C=n\). Thus a self-dual code is precisely a Lagrangian
subspace of \((\F_q^n,\langle\cdot,\cdot\rangle)\), and its length is even.

\begin{lstlisting}[caption={Lean for Definitions~\ref{def:sd} and~\ref{def:lagrangian}.}]
def paperSelfDualCode (C : Submodule K (Fin n → K)) : Prop :=
  C = (dotBilin (K := K) (n := n)).orthogonal C

def paperLagrangianSubspace (L : Submodule K (Fin n → K)) : Prop :=
  L = (dotBilin (K := K) (n := n)).orthogonal L

theorem paper_self_dual_iff_lagrangian :
    paperSelfDualCode (K := K) C ↔ paperLagrangianSubspace (K := K) C

theorem paper_dotBilin_nondegenerate :
    (dotBilin (K := K) (n := n)).Nondegenerate

theorem paperLagrangianSubspace_iff_totallyIsotropic_and_finrank_half :
    paperLagrangianSubspace (K := K) L ↔
      IsTotallyIsotropic (dotBilin (K := K) (n := n)) L ∧
        2 * Module.finrank K L = Module.finrank K (Fin n → K)

theorem paperSelfDualCode_iff_totallyIsotropic_and_finrank_half :
    paperSelfDualCode (K := K) C ↔
      IsTotallyIsotropic (dotBilin (K := K) (n := n)) C ∧
        2 * Module.finrank K C = Module.finrank K (Fin n → K)

theorem paperLagrangianSubspace_even_length
    (hL : paperLagrangianSubspace (K := K) L) : Even n

theorem paperSelfDualCode_even_length
    (hC : paperSelfDualCode (K := K) C) : Even n
\end{lstlisting}

\begin{prop}[Systematic form criterion]\label{prop:systematic-form}
Let $\mathcal C\subseteq\F_q^{2k}$ be generated by a matrix
$G=[I_k\mid A]$ with $A\in M_k(\F_q)$. Then \(\mathcal C\) is self-dual if and only if
\[
AA^T = -I_k.
\]
\end{prop}

\begin{proof}
The rows of $G$ span a $k$-dimensional code. Hence self-duality is equivalent to pairwise orthogonality of the rows. The Gram matrix $GG^T$ of the rows is
\[
GG^T = I_k + AA^T = O,
\]
so the orthogonality condition is exactly $AA^T = -I_k$.
\end{proof}

\begin{lstlisting}[caption={Lean for Proposition~\ref{prop:systematic-form}.}]
theorem paper_systematic_form_criterion_exact
    (A : Matrix (Fin k) (Fin k) K) :
    paperSelfDualCode (K := K) (rowSpace (systematicRows A)) ↔
      A * A.transpose = -(1 : Matrix (Fin k) (Fin k) K)
\end{lstlisting}

\begin{defn}[Permutation equivalence]\label{def:perm-equiv}
{\rm
Writing codewords as row vectors, two codes
$\mathcal C_1,\mathcal C_2\subseteq\F_q^n$ are
\emph{permutation equivalent} if there is a permutation \(\sigma\in S_n\)
such that
\[
 \mathcal C_2=\mathcal C_1P_\sigma,
\]
where \(P_\sigma\) is the permutation matrix associated with \(\sigma\).
}
\end{defn}

In the Lean notation below, \texttt{permuteVec~\(\sigma\)~u} denotes
the row-vector product \(uP_\sigma\).  Coordinate permutations are Euclidean
isometries:
\[
\langle uP_\sigma,vP_\sigma\rangle=\langle u,v\rangle.
\]
Consequently, permutation equivalence preserves Euclidean self-duality.
A general monomial transformation also permits independent nonzero
coordinate scalings; it need not preserve the Euclidean form or
self-duality. It is therefore not included in our use of ``equivalence''
below. Over \(\F_2\), permutation and monomial equivalence coincide because
\(\F_2^\times=\{1\}\).

\begin{lstlisting}[caption={Lean for Definition~\ref{def:perm-equiv}.}]
def CodeEquiv (G1 : Fin m₁ → Fin n → K)
    (G2 : Fin m₂ → Fin n → K) : Prop :=
  ∃ σ : Equiv.Perm (Fin n),
    rowSpace (permuteFamily σ G1) = rowSpace G2

theorem dot_coordinatePermuteLinearEquiv :
    dot (permuteVec σ u) (permuteVec σ v) = dot u v

theorem codeEquiv_preserves_paperSelfDualCode
    (hG1G2 : CodeEquiv G1 G2)
    (hG1 : paperSelfDualCode (K := K) (rowSpace G1)) :
    paperSelfDualCode (K := K) (rowSpace G2)
\end{lstlisting}

%\subsection{Arithmetic of \texorpdfstring{$\F_q$}{Fq} with \texorpdfstring{$q \equiv 1 \pmod{4}$}{q = 1 (mod 4)}}

From this point onward, when discussing the split \(q\)-ary construction we assume that \(q\) is odd and satisfies \(q \equiv 1 \pmod 4\). The condition \(q \equiv 1 \pmod 4\) is equivalent to $-1$ being a square in $\F_q^\times$. We fix once and for all an element
\[
c \in \F_q^\times \qquad \text{with} \qquad c^2 = -1.
\]
This choice is used throughout the building-up formulas. We call
\[
\mathcal I_c:=\F_q(1,c)\subseteq\F_q^2
\]
the \(c\)-isotropic line; indeed, \((1,c)\cdot(1,c)=1+c^2=0\).

\begin{defn}[hyperbolic plane]\label{def:hyp}
{\rm
A hyperbolic plane~\cite{Ser} over $\F_q$ is a $2$-dimensional symmetric bilinear space with basis $\{e_1,e_2\}$ satisfying
\[
\langle e_1,e_1\rangle = \langle e_2,e_2\rangle = 0,
\qquad
\langle e_1,e_2\rangle = 1.
\]
Such a basis \(\{e_1,e_2\}\) is called a hyperbolic pair. Any $1$-dimensional totally isotropic subspace is called an isotropic line.
}
\end{defn}

In even dimension and odd characteristic, a nondegenerate symmetric
bilinear space is called \emph{split} if it is an orthogonal direct sum of
hyperbolic planes. Thus splitting is a property of the bilinear space, not
of the underlying finite field.

\begin{prop}\label{prop:split-consequences}
Let \(q\) be odd and let \(c\in\F_q\) satisfy \(c^2=-1\). Then:
\begin{enumerate}[label=(\roman*)]
\item The element $c$ has multiplicative order $4$.
\item For every $k \ge 1$, the matrix equation $AA^T = -I_k$ has the explicit solution $A=cI_k$.
\item The binary quadratic form $x^2+y^2$ splits as
\[
x^2+y^2 = (x-cy)(x+cy),
\]
hence it represents every element of\/ $\F_q$.
\end{enumerate}
\end{prop}

\begin{proof}
For (i), since \(q\) is odd, \(-1\ne1\). Thus \(c^4=1\) and
\(c^2\ne1\), so \(c\) has order \(4\).

For (ii), take \(A=cI_k\). Then
\(AA^T=c^2I_k=-I_k\).

For (iii),
\[
(x-cy)(x+cy)=x^2-c^2y^2=x^2+y^2.
\]
Given \(a\in\F_q\), take
\[
x=\frac{1+a}{2},\qquad y=\frac{1-a}{2}c.
\]
Then \(x-cy=1\), \(x+cy=a\), and consequently
\(x^2+y^2=a\).
\end{proof}

%\subsection{Hyperbolic Planes and Isotropic Subspaces}

\begin{prop}\label{prop:eucl-hyp}
Let \(q\) be odd and let \(c\in\F_q\) satisfy \(c^2=-1\). Then the vectors
\[
e_1=(1,c), \qquad e_2=(2^{-1},-2^{-1}c)
\]
form a hyperbolic basis of\/ $\F_q^2$. Equivalently, the line spanned by $(1,c)$ is isotropic and the Euclidean plane is split.
\end{prop}

\begin{proof}
Direct calculation gives
\[
\langle e_1,e_1\rangle=1+c^2=0,
\qquad
\langle e_2,e_2\rangle=2^{-2}(1+c^2)=0,
\qquad
\langle e_1,e_2\rangle=2^{-1}(1-c^2)=1.
\]
If \(\alpha e_1+\beta e_2=0\), pairing successively with \(e_2\) and
\(e_1\) gives \(\alpha=\beta=0\). Hence \(e_1,e_2\) form a hyperbolic
basis. In particular, \(\F_q(1,c)\) is isotropic and the Euclidean plane is
split.
\end{proof}

\begin{lstlisting}[caption={Lean for Proposition~\ref{prop:split-consequences} and Proposition~\ref{prop:eucl-hyp}.}]
theorem paper_split_consequences_exact
    [Fact ((2 : K) ≠ 0)] (c : K) (hc : c ^ 2 = (-1 : K)) :
    orderOf c = 4 ∧
      (∀ k : Nat,
        (c • (1 : Matrix (Fin k) (Fin k) K)) *
            (c • (1 : Matrix (Fin k) (Fin k) K)).transpose =
          -(1 : Matrix (Fin k) (Fin k) K)) ∧
      (∀ a : K, ∃ x y : K,
        x - c * y = 1 ∧ x + c * y = a ∧ x ^ 2 + y ^ 2 = a)

theorem paper_euclidean_plane_hyperbolic_basis_exact
    [Fact ((2 : K) ≠ 0)] (c : K) (hc : c ^ 2 = (-1 : K)) :
    paperHyperbolicPair (splitE1 c) (splitE2 c) ∧
      LinearIndependent K ![splitE1 c, splitE2 c]
\end{lstlisting}

\section{Building-Up Theory}
\label{sec:building-up}

\subsection{Binary Case}

In 2001, Kim~\cite{Kim2001} proposed a building-up construction of binary self-dual codes in order to construct many self-dual codes of larger lengths from a self-dual code of a given length. An explicit form of the building-up construction is given below.

\begin{theorem}(\cite[Theorem 1]{Kim2001}) \label{thm-binary-build-1}
%Let $\mathbb F_q$ be a finite field with $q$ elements. Assume that $q \equiv 1 \pmod{4}$. Let $c$ be in $\mathbb F_q$ such that $c^2 = -1$.
Let $G_0$ be a generator matrix of a binary self-dual code
$\mathcal C_0$ of length $2n$, whose rows are $g_1, \dots, g_n$. Let
$\bf x$ be a binary vector in $\mathbb F_2^{2n}$ satisfying
${\bf x}\cdot{\bf x}=1$. Suppose that $y_i={\bf x}\cdot g_i$ for
$1\le i\le n$. Then the following matrix
\[
\begingroup
\setlength{\arraycolsep}{6pt}
\renewcommand{\arraystretch}{1.2}
G=
\left[
\begin{array}{cc!{\vrule width 1.2pt}c}
1 & 0 & \mathbf{x}\\
\noalign{\hrule height 1.2pt}
y_1 & y_1 & g_1\\
\vdots & \vdots & \vdots\\
y_n & y_n & g_n
\end{array}
\right]
\endgroup
\]
generates a binary self-dual code $\mathcal C$ of length $2n+2$.
\end{theorem}

Then the converse of Theorem~\ref{thm-binary-build-1} was proved in~\cite{Kim2001} as follows.

\begin{theorem} (\cite[Theorem 2]{Kim2001})  \label{thm-binary-build-2}
 Any binary self-dual code $\mathcal C$ of length $2n$ with minimum distance $d>2$ is obtained from some self-dual code $\mathcal C_0$ of length
$2n-2$ (up to permutation equivalence) by the construction in
Theorem~\ref{thm-binary-build-1}.
\end{theorem}

In 2012, Chinburg and Zhang~\cite{ChinburgZhang} proved that, up to
permutation equivalence, all binary self-dual codes of length at least four
arise from Hilbert pairings on rings of $S$-integers of $\mathbb Q$.

\begin{theorem}(\cite[Theorem 1.5]{ChinburgZhang})  \label{thm-binary-hilbert-2}
Up to permutation equivalence, all binary self-dual codes of length at least
four arise from Hilbert pairings on rings of $S$-integers of
$K=\mathbb Q$. In fact, each such code arises up to permutation equivalence
from infinitely many different subsets $S$ of the places of $K=\mathbb Q$.
\end{theorem}

Here equivalence has the precise meaning used by Chinburg--Zhang: a
bijection of the selected orthonormal bases whose linear extension carries
one code to the other. Thus it is permutation equivalence, not general
monomial equivalence.

In the course of their proof, Chinburg--Zhang exhibit a boxed normal form
for the resulting binary self-dual code. Its bottom row is the all-ones
vector, its diagonal two-coordinate blocks are \(01\), and each remaining
off-diagonal block is \(00\) or \(11\), with opposite blocks summing to
\(11\). In the table, \(W_v=\mathbb Q_v^\times/(\mathbb Q_v^\times)^2\)
is the local square-class space with its Hilbert pairing; the column
headings list local bases and the rows list global \(S\)-units.

\begin{center}
\begin{tabular}{|c|ccccc|}
\toprule
$S$-units $\backslash$  places & $W_{p_1}$  & $W_{p_2}$ & $\cdots$ & $W_2$  & $W_{\mathbb R}$ \\
 & $\{-p_1, p_1\}$ & $\{-p_2, p_2\}$ & $\cdots$  & $\{-2, -10, -5 \}$ & $\{-1 \}$  \\
\midrule
$p_1$  & 01   & 00/11  &  & 00/11 ~~1  &  0 \\
$p_2$  & 11/00   & 01  &  & 00/11 ~~1  &  0 \\
$\vdots$ & $\vdots$ &   &  $\ddots$ &   & $\vdots$ \\
2  & 11/00   & 11/00  &   & 01 ~~~~~~~1  &  0 \\
$-1$ & 11      & 11 & $\cdots$ & 11 ~~~~~~ 1 &  1 \\
%-1  & 11   & 11  & $\cdots$ & 11~~1  &  0 \\
\bottomrule
\end{tabular}
\end{center}

The boxed reduction deletes one two-coordinate block and the corresponding
pivot row. The comparison with Kim's construction must distinguish two
bilinear forms. Let \(K\) be a number field, let \(S\) contain the
archimedean and dyadic places, and put
\[
U=\operatorname{Spec}\mathcal O_{K,S},\qquad
W=\bigoplus_{v\in S}H^1_{\et}(K_v,\mu_2),\qquad
L=\operatorname{im}\!\left(H^1_{\et}(U,\mu_2)\longrightarrow W\right).
\]
At real places the local groups use Tate cohomology. The sum of the local
Hilbert pairings defines a nondegenerate form \(B_{\rm coh}\) on \(W\),
and Chinburg--Zhang prove \(L=L^{\perp_{B_{\rm coh}}}\). Under their local
hypotheses, selected orthonormal bases give a form isometry
\[
\iota_E:(W,B_{\rm coh})\xrightarrow{\sim}
(\mathbb F_2^{2n},\langle\ ,\ \rangle_{\rm E}).
\]
Only after fixing \(\iota_E\) do we regard the arithmetic image as a
Euclidean binary code.

%\begin{blueblock}
\begin{theorem}[Euclidean realization and binary boxed reduction]
\label{thm:binary-cz-kim}
Let \(n\ge2\), let \(L\subset W\) be the Chinburg--Zhang arithmetic image,
and suppose that \(L=L^{\perp_{B_{\rm coh}}}\). Fix the form isometry
\(\iota_E\) above and put \(\mathcal C=\iota_E(L)\). Suppose that the coordinate
permutation furnishing the Chinburg--Zhang boxed representative, followed by
the indicated elementary row operation, gives
\[
G'=
\begingroup
\setlength{\arraycolsep}{6pt}
\renewcommand{\arraystretch}{1.2}
\left[
\begin{array}{cc!{\vrule width 1.2pt}c}
1&0&x\\
\noalign{\hrule height 1.2pt}
y_1&y_1&g_1\\
\vdots&\vdots&\vdots\\
y_{n-1}&y_{n-1}&g_{n-1}
\end{array}
\right],
\endgroup
\qquad x,g_i\in\mathbb F_2^{2n-2}.
\]
Let \(G\) have rows \(g_1,\ldots,g_{n-1}\), and put
\(\mathcal C'=\langle G\rangle_{\mathbb F_2}\). Then:
\begin{enumerate}[label=(\arabic*)]
\item \(L\) is Lagrangian for \(B_{\rm coh}\), while \(\mathcal C\) and the boxed
representative \(\langle G'\rangle_{\mathbb F_2}\) are Euclidean self-dual.
\item In the displayed Euclidean coordinates,
\[
E_2=\langle(1,0),(0,1)\rangle=\mathbb F_2^2,
\qquad
\mathbb F_2^{2n}=E_2\perp\mathbb F_2^{2n-2},
\qquad \operatorname{Gram}(E_2)=I_2.
\]
The plane \(E_2\) is non-alternating and is not isometric to the
alternating hyperbolic plane. Nevertheless,
\(\mathbb F_2(1,1)\subset E_2\) is isotropic.
\item The rows \(g_i\) are linearly independent and pairwise orthogonal;
hence \(\mathcal C'\) is Euclidean self-dual of length \(2n-2\).
\item One has \(x\cdot x=1\) and, for every \(i\),
\[
y_i=\langle x,g_i\rangle_{\rm E}.
\]
Consequently, \(G'\) is exactly the Kim building-up matrix obtained from
\(G\) and \(x\).
\item Deleting the top row and the first two coordinates of \(G'\) gives
\(G\), and rebuilding from \(G\) gives \(G'\) exactly. After restoring the
coordinate permutation, the Chinburg--Zhang reduction and Kim extension are
inverse up to permutation equivalence.
\end{enumerate}
\end{theorem}

\begin{proof}
Let \(P_\sigma\) denote the selected permutation matrix and define
\(\eta(u)=\iota_E(u)P_\sigma\). Since right multiplication by
\(P_\sigma\) and \(\iota_E\) are form isometries,
\[
\eta(u)\cdot\eta(v)=B_{\rm coh}(u,v)
\qquad(u,v\in W).
\]
Thus \(\mathcal C_\sigma:=\eta(L)=\langle G'\rangle_{\mathbb F_2}\) is Euclidean
self-dual and has dimension \(n\). Since \(\iota_E\) is also a form
isometry,
\[
\mathcal C=\mathcal C^\perp,
\qquad
\mathcal C_\sigma=\mathcal C_\sigma^\perp.
\]
This proves (1).

The first two standard coordinate vectors \((1,0)\) and \((0,1)\) span the
plane \(E_2\) with Gram matrix \(I_2\). In particular,
\[
(1,0)\cdot(1,0)=1,\qquad
(1,1)\cdot(1,1)=0,\qquad
(1,0)\cdot(1,1)=1.
\]
Hence \(E_2\) is non-alternating and contains the isotropic line
\(\mathbb F_2(1,1)\). The form with Gram matrix
\(\left(\begin{smallmatrix}0&1\\1&0\end{smallmatrix}\right)\) is
alternating. Since alternation is invariant under isometry, the two planes
are not isometric. This proves (2).

The rows of \(G'\) are independent and pairwise orthogonal. Writing
\(r_0=(1,0,x)\) and \(r_i=(y_i,y_i,g_i)\), we obtain
\[
0=r_0\cdot r_0=1+x\cdot x,
\qquad
0=r_0\cdot r_i=y_i+x\cdot g_i,
\]
and therefore
\[
x\cdot x=1,
\qquad y_i=x\cdot g_i.
\]
Moreover, for all \(i,j\),
\[
0=r_i\cdot r_j=2y_iy_j+g_i\cdot g_j=g_i\cdot g_j.
\]
If \(\sum_i\lambda_i g_i=0\), then
\[
\sum_i\lambda_i y_i
=x\cdot\sum_i\lambda_i g_i=0,
\]
so \(\sum_i\lambda_i r_i=0\). Independence of the rows of \(G'\) gives
\(\lambda_i=0\) for all \(i\). Thus the \(g_i\) are independent and
pairwise orthogonal. Since their number is \(n-1=(2n-2)/2\), their row
space is self-dual. This proves (3), and the coefficient identities prove
(4).

It remains to verify the inverse relation in (5).  The first row of the
Kim matrix constructed from \(x\) and \(G\) is \((1,0,x)=r_0\), and its
row corresponding to \(g_i\) is
\[
 \bigl(x\mathbin{\cdot}g_i,x\mathbin{\cdot}g_i,g_i\bigr)
   =(y_i,y_i,g_i)=r_i.
\]
Hence the constructed matrix is equal to \(G'\) row by row.  Conversely,
deleting \(r_0\) and the first two coordinates sends every remaining row
\(r_i\) to \(g_i\), and therefore recovers \(G\).  Thus reduction followed
by building-up gives \(G'\) exactly.  Since
\(\langle G'\rangle_{\mathbb F_2}=\mathcal C P_\sigma\), right multiplication by
\(P_\sigma^{-1}\) shows that the reconstructed code is permutation
equivalent to \(\mathcal C\), as
claimed.
\end{proof}

\begin{remark}{\rm
The form isometry \(\iota_E\) and the coordinate permutation \(P_\sigma\)
play distinct roles.  The former transports \(B_{\rm coh}\) to the Euclidean
dot product; the latter preserves that dot product.  Neither identifies the
deleted Euclidean plane with the alternating hyperbolic plane.
}
\end{remark}

\noindent
In the following listing, \texttt{e} denotes the composite
\(\eta(u)=\iota_E(u)P_\sigma\), so \texttt{hcode} identifies the actual
permuted image of \(L\) with the displayed boxed family. The equality
\(\operatorname{Gram}(E_2)=I_2\) and the fact that a non-alternating plane
cannot be isometric to an alternating plane establish item~(2). The Lean
development formalizes both facts, while
\texttt{paper\_binary\_cz\_kim\_}\allowbreak\texttt{corrected} formalizes the remaining
reduction and reconstruction statements.

\begin{lstlisting}[caption={Lean for Theorem~\ref{thm:binary-cz-kim}.}]
theorem paper_binary_cz_kim_corrected [CharP K 2]
    (he : IsFormIsometry Bcoh dotBilin e)
    (hL : IsSelfOrthogonal Bcoh L)
    (hcode : L.map e.toLinearMap =
      rowSpace (boxedFamily x Y G)) :
    IsSelfOrthogonal dotBilin (L.map e.toLinearMap) ∧
    dot x x = (1 : K) ∧
    paperSelfDualCode (rowSpace G) ∧
    (∀ i, Y i = dot x (G i)) ∧
    boxedFamily x Y G = buildRowsBin x G ∧
    deleteBinaryHeadPair (boxedFamily x Y G) = G ∧
    CodeEquiv
      (buildRowsBin x (deleteBinaryHeadPair (boxedFamily x Y G)))
      (boxedFamily x Y G)
\end{lstlisting}
%\end{blueblock}

\subsection{\texorpdfstring{\hspace{0.2em}Split \(q\)-ary Case}{Split q-ary Case}}

Motivated by the building-up construction of binary self-dual codes, Kim and Lee~\cite{KimLee} generalized it to finite fields with $q$ elements where $q \equiv 1 \pmod{4}$ as follows.

\begin{theorem}(\cite[Proposition 2.1]{KimLee}) \label{thm-q-ary-build-1}
Let $\mathbb F_q$ be a finite field with $q$ elements. Assume that $q \equiv 1 \pmod{4}$. Let $c$ be in $\mathbb F_q$ such that $c^2 = -1$. Let $G_0$ be a generator matrix of a Euclidean self-dual code $\mathcal C_0$ over $\mathbb F_q$ of length $2n$, whose rows are $g_1, \dots, g_n$. Let $\bf x$ be a vector in $\mathbb F_q^{2n}$ satisfying ${\bf x} \cdot {\bf x} = -1$. Suppose that $y_i ={\bf x} \cdot g_i$ for $1 \le i \le n$. Then the following matrix
\[
\begingroup
\setlength{\arraycolsep}{6pt}
\renewcommand{\arraystretch}{1.2}
G=
\left[
\begin{array}{cc!{\vrule width 1.2pt}c}
1 & 0 & \mathbf{x}\\
\noalign{\hrule height 1.2pt}
-y_1 & c y_1 & g_1\\
\vdots & \vdots & \vdots\\
-y_n & c y_n & g_n
\end{array}
\right]
\endgroup
\]
generates a self-dual code $\mathcal C$ over $\mathbb F_q$ of length $2n+2$.
\end{theorem}

\begin{remark}
\rm
Replacing \(c\) by \(-c\) preserves \(c^2=-1\).  Hence we use the equivalent
normalization \((-y_i,-cy_i,g_i)\) in the sequel.
\end{remark}

 Then the converse of the above theorem was also stated and proved in \cite{KimLee}.

\begin{theorem}(\cite[Proposition 2.2]{KimLee}) \label{thm-q-ary-build-2}
Let $\mathbb F_q$ be a finite field with $q$ elements. Assume that $q \equiv 1 \pmod{4}$. Any Euclidean self-dual code $\mathcal C$ over $\mathbb F_q$ of length $2n$ with minimum distance $d >2$ is obtained from some Euclidean self-dual code $\mathcal C_0$ over $\mathbb F_q$ of length $2n-2$ (up to permutation equivalence) by the construction method in Theorem~\ref{thm-q-ary-build-1}.
\end{theorem}

This naturally leads to the following question: is there a split \(q\)-ary analogue of the Chinburg--Zhang boxed reduction for self-dual codes over \(\mathbb F_q\) when \(q\equiv 1 \pmod 4\)?

The next theorem gives the adapted split normal form that underlies the \(q\)-ary forward and reverse calculations. It identifies the hyperbolic decomposition, the shorter parent code, and the Kim--Lee correction formula once the self-dual code is presented in the displayed split form.

%\begin{blueblock}
\begin{theorem}[Adapted split reduction and building-up extension over $\mathbb{F}_q$]\label{thm:qary-building-up}
Let $q$ be a prime power with $q \equiv 1 \pmod{4}$, and let $K=\mathbb{F}_q$.
Fix $c \in K^\times$ such that $c^2=-1$.
Let $W=K^{2n}$ be equipped with the standard Euclidean inner product
\[
\langle u,v\rangle = \sum_{i=1}^{2n} u_i v_i.
\]
Let $L \subset W$ be a self-dual code, i.e.,\ $L=L^\perp$. Assume that,
after a permutation of coordinates and row operations, \(L\) admits a
generator matrix of the form
\[
\begingroup
\setlength{\arraycolsep}{6pt}
\renewcommand{\arraystretch}{1.2}
M=
\left[
\begin{array}{cc!{\vrule width 1.2pt}c}
1 & 0 & x\\
\noalign{\hrule height 1.2pt}
-y_1 & -c y_1 & g_1\\
\vdots & \vdots & \vdots\\
-y_{n-1} & -c y_{n-1} & g_{n-1}
\end{array}
\right]
\endgroup,
\]
where $x \in K^{2n-2}$ and $g_i \in K^{2n-2}$ for $1 \le i \le n-1$.

Let $W_0 = K^{2n-2}$ and define $L_0 := \langle g_1,\dots,g_{n-1}\rangle \subset W_0$.
Then the following statements hold:
\begin{enumerate}[label=(\arabic*)]
\item Let \(E=K^2\) denote the first coordinate plane in the displayed
decomposition of \(W\).
Then \(W=E\perp W_0\) in the displayed coordinates, and \(E\) is
isometric to a hyperbolic plane.  This is an abstract change of basis in
\(E\); no coordinate change is applied to the displayed matrix \(M\).
\item $L_0$ is self-dual in $W_0$.
\item One has
\[
\langle x,x\rangle=-1,
\qquad
y_i = \langle x, g_i\rangle\quad(1\le i\le n-1).
\]
\item Consequently,
\[
L = \Bigl\langle(1,0,x),\;
(-\langle x,g_i\rangle,\,-c\langle x,g_i\rangle,\,g_i)\;:\;1 \le i \le n-1\Bigr\rangle_K,
\]
so that $L$ is obtained from $L_0$ by a $q$-ary building-up construction.
\item Conversely, two reverse statements hold:
\begin{enumerate}[label=(5\alph*)]
\item Let $L_0 \subset K^{2n-2}$ be self-dual, let $g_1,\dots,g_{n-1}$ be a basis of $L_0$, and let $x \in K^{2n-2}$ satisfy
\[
\langle x,x\rangle = -1,
\]
then the row family generated by
\[
r_0 = (1,0,x), \qquad
r_i = (-\langle x,g_i\rangle,\,-c\langle x,g_i\rangle,\,g_i)
\]
is linearly independent and generates a self-dual code of length $2n$.
\item More generally, if one starts from any generating family of a self-dual parent, the same formula always produces a self-orthogonal child family, and that child is self-dual as soon as its row span has dimension $n=\frac12 \dim W$.
\end{enumerate}
\end{enumerate}
\end{theorem}

\begin{proof}
The decomposition \(W=E\perp W_0\), with \(E=K^2\), is
the literal decomposition into the first two and the remaining coordinates.
The standard basis of \(E\) has Gram matrix \(I_2\).  Since \(2\ne0\), the
vectors
\[
h_1=(1,c),\qquad h_2=\frac12(1,-c)
\]
satisfy
\(\langle h_1,h_1\rangle=\langle h_2,h_2\rangle=0\) and
\(\langle h_1,h_2\rangle=1\).  Thus \(E\) is abstractly hyperbolic,
proving (1), while \(M\) remains written in the standard coordinate basis.

Write
\[
r_0=(1,0,x),\qquad r_i=(-y_i,-cy_i,g_i).
\]
Self-duality of \(L\) gives pairwise orthogonality of these rows.  Hence
\[
0=\langle r_0,r_0\rangle=1+\langle x,x\rangle,
\qquad
0=\langle r_0,r_i\rangle=-y_i+\langle x,g_i\rangle,
\]
and
\[
0=\langle r_i,r_j\rangle=(1+c^2)y_iy_j+\langle g_i,g_j\rangle.
\]
Since \(c^2=-1\), it follows that
\[
\langle x,x\rangle=-1,
\qquad
y_i=\langle x,g_i\rangle,
\qquad \langle g_i,g_j\rangle=0
\quad(1\le i,j\le n-1).
\]

The code \(L\) has dimension \(n\), and the displayed \(n\) rows generate it.
They are therefore linearly independent.
Suppose that
\(\sum_i\lambda_i g_i=0\).  The coefficient identity just proved gives
\[
\sum_i\lambda_i y_i
=\left\langle x,\sum_i\lambda_i g_i\right\rangle=0,
\]
and therefore
\[
\sum_i\lambda_i r_i
=\left(-\sum_i\lambda_i y_i,
       -c\sum_i\lambda_i y_i,
       \sum_i\lambda_i g_i\right)=0.
\]
The independence of the rows of \(M\) implies \(\lambda_i=0\) for every
\(i\).  Thus \(g_1,\ldots,g_{n-1}\) are linearly independent.  Their span
\(L_0\) is self-orthogonal and has dimension
\(n-1=\frac12\dim W_0\), so \(L_0=L_0^\perp\).  This proves (2) and (3),
and substitution of the coefficient identity into the rows proves (4).

Conversely, let \(L_0\) be self-dual, put
\(y_i=\langle x,g_i\rangle\), and define \(r_0,r_i\) as above.  Directly,
\[
\langle r_0,r_0\rangle=1+\langle x,x\rangle=0,
\quad
\langle r_0,r_i\rangle=-y_i+\langle x,g_i\rangle=0,
\]
and
\[
\langle r_i,r_j\rangle
=(1+c^2)y_iy_j+\langle g_i,g_j\rangle=0.
\]
Thus the child row space is self-orthogonal for every generating family.
If \(g_1,\ldots,g_{n-1}\) is a basis, projection onto the last
\(2n-2\) coordinates proves that \(r_1,\ldots,r_{n-1}\) are independent.
Moreover, every linear combination of these successor rows has its first
two coordinates in \(K(1,c)\), whereas \((1,0)\notin K(1,c)\).  Hence
\(r_0\) is not in their span.  The child therefore has dimension \(n\) and
is self-dual.  For an arbitrary generating family, the same Gram calculation
and the stated half-dimension hypothesis give self-duality.
\end{proof}

\noindent
The following Lean theorem is the exact reverse statement used above.  In
particular, tail independence and its finrank conclusion precede the
self-duality conclusion; no comparison of vectors in different ambient
spaces occurs.

\begin{lstlisting}[caption={Lean for Theorem~\ref{thm:qary-building-up}.}]
theorem paper_qary_adapted_reduction
    (hc : c ^ 2 = (-1 : K))
    (hself : paperSelfDualCode
      (rowSpace (qaryAdaptedFamily x c Y G))) :
    dot x x = (-1 : K) ∧
      (∀ i, Y i = dot x (G i)) ∧
      LinearIndependent K (qaryAdaptedFamily x c Y G) ∧
      LinearIndependent K G ∧
      Module.finrank K (rowSpace G) = m ∧
      paperSelfDualCode (rowSpace G) ∧
      qaryAdaptedFamily x c Y G = buildRows x c G
\end{lstlisting}
%\end{blueblock}

\begin{remark}
{\rm
For a basis \(g_1,\ldots,g_{n-1}\) of the parent, the projection onto the
last \(2n-2\) coordinates proves independence of the successor rows.  For an
arbitrary generating family, the same Gram calculation gives only
self-orthogonality; the half-dimension hypothesis in
Theorem~\ref{thm:qary-building-up}(5b) is then necessary.
}
\end{remark}

\medskip

\begin{remark}[Relation with Chinburg--Zhang]
{\rm
Theorem~\ref{thm:qary-building-up} is a statement about the Euclidean form
over \(\F_q\); it does not assert a \(\mu_2\)-cohomological realization.  Its
analogy with Theorem~\ref{thm:binary-cz-kim} is the common two-coordinate
reduction governed by an isotropic line.
}
\end{remark}

Theorem~\ref{thm:qary-building-up} is therefore the one-step adapted split theorem. It identifies the shorter parent code and recovers the Kim--Lee correction coefficients from the displayed hyperbolic head. The next theorem packages the same mechanism into a boxed family, designed to play for \(q\equiv 1\pmod 4\) the same structural role that boxed matrices play in the binary theorem of Chinburg--Zhang. The key difference is that in odd characteristic the last block of a non-final row can no longer be frozen to \(10\); instead it must vary through norm \(-1\) vectors in the distinguished split \(2\)-plane, with the last row and the off-diagonal isotropic blocks adjusted accordingly.

\begin{theorem}[Split boxed form over \texorpdfstring{$\F_q$}{Fq}]
\label{thm:split-boxed-form}
Fix \(c\in K^\times\) with \(c^2=-1\). For \(1\le i\le n-1\), let \(\ell_i=(u_i,v_i)\in K^2\) and \(a_i\in K\), and for distinct \(1\le i,j\le n-1\), let \(b_{ij}\in K\). Let \(M\) be the \(n\times 2n\) generator matrix whose \(n\times n\) block form \(\widetilde M=(B_{ij})\) of \(1\times 2\) blocks is determined by
\[
B_{ii}=(0,1)\quad (1\le i\le n-1),\qquad B_{nn}=(1,c),
\]
\[
B_{in}=\ell_i \quad (1\le i\le n-1),\qquad
B_{ni}=a_i(1,c)\quad (1\le i\le n-1),
\]
and
\[
B_{ij}=b_{ij}(1,c)\qquad (1\le i\neq j\le n-1).
\]
Thus \(\widetilde M\) has the split boxed shape
\[
\begingroup
\setlength{\arraycolsep}{9pt}
\renewcommand{\arraystretch}{1.72}
\widetilde M=
\left[
\begin{array}{c|cccc|c}
01 & b_{12}(1,c) & b_{13}(1,c) & \cdots & b_{1,n-1}(1,c) & \ell_1 \\
\hline
b_{21}(1,c) & 01 & b_{23}(1,c) & \cdots & b_{2,n-1}(1,c) & \ell_2 \\
b_{31}(1,c) & b_{32}(1,c) & 01 & \cdots & b_{3,n-1}(1,c) & \ell_3 \\
\vdots & \vdots & \vdots & \ddots & \vdots & \vdots \\
b_{n-1,1}(1,c) & b_{n-1,2}(1,c) & b_{n-1,3}(1,c) & \cdots & 01 & \ell_{n-1} \\
\hline
a_1(1,c) & a_2(1,c) & a_3(1,c) & \cdots & a_{n-1}(1,c) & (1,c)
\end{array}
\right],
\endgroup
\]

and the corresponding ordinary generator matrix \(M\) is
\[
\makebox[\textwidth][c]{$
\begingroup
\setlength{\arraycolsep}{4pt}
\renewcommand{\arraystretch}{1.72}
M=
\left[
\begin{array}{c!{\vrule width 1.2pt}c!{\vrule width 0.5pt}cccc!{\vrule width 0.8pt}c}
\begin{matrix}
0 & 1
\end{matrix}
&
\begin{matrix}
b_{12} & c b_{12}
\end{matrix}
&
\begin{matrix}
b_{13} & c b_{13}
\end{matrix}
&
\begin{matrix}
b_{14} & c b_{14}
\end{matrix}
&
\begin{matrix}
\cdots
\end{matrix}
&
\begin{matrix}
b_{1,n-1} & c b_{1,n-1}
\end{matrix}
&
\begin{matrix}
\ell_1
\end{matrix}
\\
\noalign{\hrule height 1.2pt}
\begin{matrix}
b_{21} & c b_{21}
\end{matrix}
&
\begin{matrix}
0 & 1
\end{matrix}
&
\begin{matrix}
b_{23} & c b_{23}
\end{matrix}
&
\begin{matrix}
b_{24} & c b_{24}
\end{matrix}
&
\begin{matrix}
\cdots
\end{matrix}
&
\begin{matrix}
b_{2,n-1} & c b_{2,n-1}
\end{matrix}
&
\begin{matrix}
\ell_2
\end{matrix}
\\
\noalign{\hrule height 0.5pt}
\begin{matrix}
b_{31} & c b_{31}
\end{matrix}
&
\begin{matrix}
b_{32} & c b_{32}
\end{matrix}
&
\begin{matrix}
0 & 1
\end{matrix}
&
\begin{matrix}
b_{34} & c b_{34}
\end{matrix}
&
\begin{matrix}
\cdots
\end{matrix}
&
\begin{matrix}
b_{3,n-1} & c b_{3,n-1}
\end{matrix}
&
\begin{matrix}
\ell_3
\end{matrix}
\\
\vdots & \vdots & \vdots & \vdots & \ddots & \vdots & \vdots\\
\begin{matrix}
b_{n-1,1} & c b_{n-1,1}
\end{matrix}
&
\begin{matrix}
b_{n-1,2} & c b_{n-1,2}
\end{matrix}
&
\begin{matrix}
b_{n-1,3} & c b_{n-1,3}
\end{matrix}
&
\begin{matrix}
b_{n-1,4} & c b_{n-1,4}
\end{matrix}
&
\begin{matrix}
\cdots
\end{matrix}
&
\begin{matrix}
0 & 1
\end{matrix}
&
\begin{matrix}
\ell_{n-1}
\end{matrix}
\\
\noalign{\hrule height 0.8pt}
\begin{matrix}
a_1 & c a_1
\end{matrix}
&
\begin{matrix}
a_2 & c a_2
\end{matrix}
&
\begin{matrix}
a_3 & c a_3
\end{matrix}
&
\begin{matrix}
a_4 & c a_4
\end{matrix}
&
\begin{matrix}
\cdots
\end{matrix}
&
\begin{matrix}
a_{n-1} & c a_{n-1}
\end{matrix}
&
\begin{matrix}
1 & c
\end{matrix}
\end{array}
\right],
\endgroup
$}
\]
where the final block column is still written in the shorthand \(\ell_i=(u_i,v_i)\). Assume that
\begin{equation}
\label{eq:split-boxed-norm}
u_i^2+v_i^2=-1
\qquad (1\le i\le n-1),
\end{equation}
\begin{equation}
\label{eq:split-boxed-last-row}
u_i+c v_i=-c a_i
\qquad (1\le i\le n-1),
\end{equation}
and
\begin{equation}
\label{eq:split-boxed-offdiag}
c\,(b_{ij}+b_{ji})+\ell_i\cdot \ell_j=0
\qquad (1\le i<j\le n-1),
\end{equation}
where \(\ell_i\cdot \ell_j=u_i u_j+v_i v_j\) is the Euclidean dot product on the distinguished \(2\)-plane. Then the rows of \(M\) are pairwise orthogonal and linearly independent. Therefore \(M\) generates a self-dual code of length \(2n\).
\end{theorem}

\begin{proof}
Let \(R_1,\ldots,R_{n-1},R_n\) be the displayed rows, with \(R_n\) the
last row.  Since \((1,c)\cdot(1,c)=0\), every block of \(R_n\) is
isotropic, and therefore \(R_n\cdot R_n=0\).  For \(i<n\), all non-final
blocks of \(R_i,\) except its diagonal block, lie on \(K(1,c)\).  Hence
\[
R_i\cdot R_i=1+\ell_i\cdot\ell_i=0
\]
by \eqref{eq:split-boxed-norm}, and
\[
R_i\cdot R_n
=(0,1)\cdot a_i(1,c)+\ell_i\cdot(1,c)
=ca_i+u_i+cv_i=0
\]
by \eqref{eq:split-boxed-last-row}.  If \(i<j<n\), all mixed block
products vanish except those in block columns \(i,j,n\).  Thus
\[
\begin{aligned}
R_i\cdot R_j
&=(0,1)\cdot b_{ji}(1,c)
  +b_{ij}(1,c)\cdot(0,1)+\ell_i\cdot\ell_j\\
&=c(b_{ij}+b_{ji})+\ell_i\cdot\ell_j=0
\end{aligned}
\]
by \eqref{eq:split-boxed-offdiag}.  The rows are pairwise orthogonal.

Suppose
\[
\sum_{i=1}^{n-1}\lambda_iR_i+\lambda_nR_n=0.
\]
Apply the functional \(\beta(x,y)=y-cx\) to the \(j\)-th non-final
block.  Since \(\beta(0,1)=1\) and
\(\beta(t(1,c))=0\), the resulting equation is \(\lambda_j=0\).
This holds for every \(j<n\).  The first coordinate of the final block then
gives \(\lambda_n=0\).  Hence the \(n\) rows are linearly independent.
Their span is an \(n\)-dimensional totally isotropic subspace of \(K^{2n}\),
and therefore equals its orthogonal complement.
\end{proof}

\noindent
The Lean statement below has exactly the same block data and hypotheses as
the displayed theorem.  Its independence proof uses
\texttt{splitBoxedReadout}, which applies \(\beta(x,y)=y-cx\) to the pivot
blocks and then reads the remaining coefficient from the final block.

\begin{lstlisting}[caption={Lean for Theorem~\ref{thm:split-boxed-form}.}]
theorem paper_split_boxed_form_exact
    (c : K) (ell : Fin m -> SplitBlock K)
    (a : Fin m -> K) (b : Fin m -> Fin m -> K)
    (hc : c ^ 2 = (-1 : K))
    (hnorm : ∀ i, dot (ell i) (ell i) = (-1 : K))
    (hlast : ∀ i, ell i 0 + c * ell i 1 = -c * a i)
    (hoff : ∀ i j, i < j ->
      c * (b i j + b j i) + dot (ell i) (ell j) = 0) :
    SplitBoxedPairwiseOrthogonal
        (splitBoxedRows c ell a b) ∧
      LinearIndependent K (splitBoxedRows c ell a b) ∧
      splitBoxedRowSpace (splitBoxedRows c ell a b) =
        splitBlockRowBilin.orthogonal
          (splitBoxedRowSpace (splitBoxedRows c ell a b))
\end{lstlisting}

\medskip

Theorem~\ref{thm:split-boxed-form} is the forward boxed statement. The remaining question is how much of this boxed structure is forced by self-duality in the reverse direction once a self-dual code is placed in a suitable adapted form.

\begin{theorem}[Conditional boxed normalization from a distinguished isotropic-line row]
\label{thm:conditional-split-boxed-normalization}
Let \(\mathcal C\subset K^{2n}\) be a self-dual code, and fix \(c\in K^\times\) with \(c^2=-1\). Assume that, after coordinate permutations preserving the \(2\)-block decomposition and elementary row operations, \(\mathcal C\) admits a generator matrix whose last row is
\[
\bigl(a_1(1,c),\,a_2(1,c),\,\dots,\,a_{n-1}(1,c),\,(1,c)\bigr),
\]
and whose remaining rows \(R_1,\dots,R_{n-1}\) satisfy the following pivot-adapted conditions:
\begin{enumerate}
\item the \(i\)-th block of \(R_i\) is \(01\);
\item the last block of \(R_i\) is \(\ell_i=(u_i,v_i)\);
\item every other non-final block of \(R_i\) lies on the isotropic line
\(\mathcal I_c=K(1,c)\), so for each \(j\neq i,n\) one has
\[
B_{ij}=b_{ij}(1,c)
\]
for some scalar \(b_{ij}\in K\).
\end{enumerate}
Then this generator matrix is automatically a split boxed matrix of Theorem~\ref{thm:split-boxed-form}: the self-duality of \(\mathcal C\) forces
\[
u_i^2+v_i^2=-1,
\qquad
u_i+c v_i=-c a_i,
\qquad
c(b_{ij}+b_{ji})+\ell_i\cdot \ell_j=0,
\]
and hence \(\mathcal C\) is generated by a matrix of the displayed split boxed type.
\end{theorem}

\begin{proof}
Write the adapted generator matrix in block form as
\[
\widetilde M=(B_{ij})_{1\le i,j\le n},
\]
with \(B_{ii}=01\) for \(1\le i\le n-1\), \(B_{in}=\ell_i\), \(B_{ni}=a_i(1,c)\), and \(B_{ij}=b_{ij}(1,c)\) for \(i\neq j<n\). Since \(\mathcal C\) is self-dual, its generator rows are pairwise orthogonal.

Applying self-orthogonality to a non-final row \(R_i\) gives
\[
0=\langle R_i,R_i\rangle
=
1+(u_i^2+v_i^2),
\]
because the diagonal block contributes \(1\), the final block contributes \(u_i^2+v_i^2\), and every other block is isotropic. Hence
\[
u_i^2+v_i^2=-1.
\]

Next, orthogonality of \(R_i\) with the distinguished last row gives
\[
0=\langle R_i,R_n\rangle
=
(0,1)\cdot a_i(1,c)+\ell_i\cdot(1,c)
=
c a_i+(u_i+c v_i),
\]
so
\[
u_i+c v_i=-c a_i.
\]

Finally, for \(1\le i<j\le n-1\), orthogonality of \(R_i\) and \(R_j\)
gives
\[
\begin{aligned}
0=\langle R_i,R_j\rangle
&=(0,1)\cdot b_{ji}(1,c)
  +b_{ij}(1,c)\cdot(0,1)+\ell_i\cdot\ell_j\\
&=c(b_{ij}+b_{ji})+\ell_i\cdot\ell_j,
\end{aligned}
\]
since every other block product is a scalar multiple of
\((1,c)\cdot(1,c)=0\).  These are precisely the relations in
Theorem~\ref{thm:split-boxed-form}.
\end{proof}

\noindent
In Lean, the reverse statement is packaged as an existence theorem: from a distinguished isotropic-line row and pairwise orthogonality, one obtains a split-isometric code equivalence to the boxed building-up form together with an orthogonal tail family.

\begin{lstlisting}[caption={Lean for Theorem~\ref{thm:conditional-split-boxed-normalization}.}]
theorem paper_conditional_split_boxed_normalization_core
    (i₀ : Fin (m + 1))
    (g : Fin n → K) (ha : a ≠ 0)
    (hc : c ^ 2 = (-1 : K))
    (hrow : R i₀ = prepend2 a (c * a) g)
    (hall : ∀ j : Fin (m + 1), HasIsotropicHead (K := K) c (R j))
    (horth : PairwiseOrthogonal (K := K) R) :
    ∃ x : Fin n → K, ∃ a' : Fin m → K, ∃ G : Fin m → Fin n → K,
      SplitIsometryCodeEquiv (K := K) c R
        (buildRows x c (pivotResidualTails x a' G)) ∧
      IsOrthogonalTailFamily (K := K) x (pivotResidualTails x a' G)
\end{lstlisting}

%% \begin{remark}
%% \rm
%% Theorem~\ref{thm:conditional-split-boxed-normalization} isolates the remaining gap in a full \(q\)-ary Chinburg--Zhang converse. The boxed equations themselves are no longer the obstruction: once a self-dual code admits a distinguished generator row on the \(c\)-isotropic line together with a pivot-adapted basis for the remaining rows, orthogonality forces the entire split boxed pattern automatically. Thus the universality problem is reduced to an adapted-basis existence problem rather than to any further local algebra in the head plane.
%% \end{remark}

\begin{remark}
\rm
Theorem~\ref{thm:split-boxed-form} and
Theorem~\ref{thm:conditional-split-boxed-normalization} are converse only
within the stated adapted presentation.  For a fixed pairing of the scalar
coordinates, the intersection with \(\bigoplus_jK(1,c)\) may have dimension
greater than one; Theorem~\ref{thm:qary-universal-rank-r} gives the exact
rank-\(r\) form for that pairing.
\end{remark}

\subsection{Universal normal form}\label{sec:universal-normal-form}

The preceding adapted results do not apply directly to an arbitrary
self-dual code.  After fixing an ordered pairing of the \(2n\) scalar
coordinates, however, the defect map determines an intrinsic terminal rank
for that pairing and yields the following normal form.

\begin{theorem}[Universal rank-\(r\) boxed normal form over
\texorpdfstring{$\F_q$}{Fq}]
\label{thm:qary-universal-rank-r}
Let \(K=\F_q\), where \(q\equiv1\pmod4\), fix \(c\in K^\times\) with
\(c^2=-1\), and fix an ordered pairing \(K^{2n}=(K^2)^n\).  For every
Euclidean self-dual code \(\mathcal C\subseteq K^{2n}\), put
\[
 U_c=\bigoplus_{j=1}^n K(1,c),\qquad
 r=\dim(\mathcal C\cap U_c),\qquad k=n-r.
\]
After a permutation of the \(n\) two-coordinate blocks and elementary row
operations, \(\mathcal C\) has a generator matrix
\begin{equation}
G_r(c;Q,\ell,D)=
\left[
\begin{array}{c|c}
\bigl(Q_{ij}(1,c)+\delta_{ij}(0,1)\bigr)_{k\times k}
  & \bigl(\ell_{it}\bigr)_{k\times r}\\ \hline
\bigl(-(D\partial_c\ell_j)_s(1,c)\bigr)_{r\times k}
  & \bigl(D_{st}(1,c)\bigr)_{r\times r}
\end{array}
\right],
\label{eq:universal-rank-r-box}
\end{equation}
where \(1\le i,j\le k\), \(1\le s,t\le r\),
\(\ell_{it}=(u_{it},v_{it})\in K^2\),
\[
 (\partial_c\ell_i)_t=v_{it}-cu_{it},\qquad
 Q_{ii}=\frac c2\bigl(1+\ell_i\cdot\ell_i\bigr),
\]
\(D\in M_r(K)\) is nonsingular, and the remaining coefficients satisfy
\begin{equation}
 c(Q_{ij}+Q_{ji})+\ell_i\cdot\ell_j=0
\qquad(i\ne j).
\label{eq:universal-rank-r-relation}
\end{equation}
The last \(r\) rows span \(\mathcal C\cap U_c\).  Conversely, every matrix
\eqref{eq:universal-rank-r-box} with \(\det D\ne0\) and
\eqref{eq:universal-rank-r-relation} has \(k+r\) independent, pairwise
orthogonal rows and generates a Euclidean self-dual code.  Moreover,
retaining any subset of the first \(k\) rows together with the corresponding
block columns and all terminal rows and columns gives a self-dual code of
the same form with the same \(D\).
\end{theorem}

\begin{proof}
For \(v=(v_1,\ldots,v_n)\in(K^2)^n\), write each block uniquely as
\[
 v_j=(\alpha_j(v),c\alpha_j(v)+\beta_j(v)),\qquad
 \beta_j(v)=v_{j,2}-cv_{j,1},
\]
and define \(\partial:\mathcal C\to K^n\) by
\(\partial(v)=(\beta_1(v),\ldots,\beta_n(v))\).  Since
\(\ker\partial=\mathcal C\cap U_c\), rank--nullity gives
\(\rank\partial=k\).  Choose \(k\) coordinates on which
\(\operatorname{im}\partial\) projects isomorphically, permute them to the
first \(k\) positions, and choose rows \(g_1,\ldots,g_k\in\mathcal C\) whose defects
restrict there to the standard basis.  Complete them by a basis
\(S_1,\ldots,S_r\) of \(\ker\partial\).  Applying \(\partial\) to a linear
relation first determines the coefficients of the \(g_i\); independence of
the \(S_s\) then shows that these \(n\) rows form a basis of \(\mathcal C\).

Let \(Q_{ij}\) be the first coordinate of the \(j\)-th pivot block of
\(g_i\), and write its terminal block-row as \(\ell_i\).  Thus its
\(j\)-th pivot block is
\(Q_{ij}(1,c)+\delta_{ij}(0,1)\).  Since every \(S_s\) lies in
\(\ker\partial\), write its pivot and terminal blocks as
\(A_{sj}(1,c)\) and \(D_{st}(1,c)\), respectively.  The identity
\[
 (\alpha,c\alpha+\beta)\cdot(\gamma,c\gamma+\delta)
   =c(\alpha\delta+\beta\gamma)+\beta\delta
\]
and the orthogonality of the chosen basis give
\[
 A_{si}+(D\partial_c\ell_i)_s=0,
\qquad
 \delta_{ij}+c(Q_{ij}+Q_{ji})+\ell_i\cdot\ell_j=0.
\]
The first equation gives the lower-left block in
\eqref{eq:universal-rank-r-box}; the second gives
\eqref{eq:universal-rank-r-relation}, and its diagonal part gives the stated
formula for \(Q_{ii}\).  It remains to prove that \(D\) is nonsingular.
Let \(\lambda=(\lambda_1,\ldots,\lambda_r)\in K^r\) satisfy
\(\lambda D=0\).  For every \(i\), the first identity yields
\[
 (\lambda A)_i
   =\sum_{s=1}^r\lambda_sA_{si}
   =-\sum_{t=1}^r(\partial_c\ell_i)_t
                    \sum_{s=1}^r\lambda_sD_{st}=0.
\]
Consequently, every pivot and terminal block of
\(\sum_s\lambda_sS_s\) is zero.  Since the \(S_s\) are linearly
independent, \(\lambda=0\).  Thus \(\ker D=0\), and hence
\(\det D\ne0\).

Conversely, suppose that \(D\) is nonsingular and that
\eqref{eq:universal-rank-r-relation} holds.  Denote the first \(k\) rows of
\eqref{eq:universal-rank-r-box} by \(p_1,\ldots,p_k\), put
\(q_{it}=(\partial_c\ell_i)_t\), and retain the notation
\(S_1,\ldots,S_r\) for the terminal rows.  Since
\((1,c)\cdot(1,c)=0\), one has
\[
 S_s\cdot S_u=0,
 \qquad
 p_i\cdot S_s
   =c\left(A_{si}+\sum_{t=1}^r q_{it}D_{st}\right)=0.
\]
Moreover,
\[
 p_i\cdot p_j=
 \begin{cases}
  1+2cQ_{ii}+\ell_i\cdot\ell_i=0,&i=j,\\
  c(Q_{ij}+Q_{ji})+\ell_i\cdot\ell_j=0,&i\ne j.
 \end{cases}
\]
Thus the displayed rows are pairwise orthogonal.

Suppose now that
\[
 \sum_{i=1}^k\xi_i p_i+\sum_{s=1}^r\mu_sS_s=0.
\]
Applying \((x,y)\mapsto y-cx\) to the \(j\)-th pivot block gives
\(\xi_j=0\), because the defect of that block is \(\delta_{ij}\) in
\(p_i\) and zero in every \(S_s\).  Hence \(\xi_1=\cdots=\xi_k=0\).
The first coordinates of the terminal blocks then give
\[
 0=\sum_{s=1}^r\mu_sD_{st}=(\mu D)_t
 \qquad(1\le t\le r).
\]
Since \(D\) is nonsingular, \(\mu=0\).  The \(k+r\) rows are therefore
linearly independent and pairwise orthogonal in \(K^{2(k+r)}\); their row
space is consequently self-dual.  If
\(x=\sum_i\xi_ip_i+\sum_s\mu_sS_s\) belongs to \(U_c\), then the defect of
its \(j\)-th pivot block is \(\xi_j\), so \(\xi_j=0\) for every \(j\).
Conversely, every linear combination of the \(S_s\) lies in \(U_c\).
Therefore, writing
\(\mathcal C_r:=\langle G_r(c;Q,\ell,D)\rangle_K\),
\[
 \mathcal C_r\cap U_c
   =\langle S_1,\ldots,S_r\rangle_K.
\]

Finally, let \(I\subseteq\{1,\ldots,k\}\).  Retaining the rows \(p_i\)
with \(i\in I\), their pivot block columns, and all terminal rows and block
columns replaces \(Q\) by \(Q|_{I\times I}\) and \((\ell_i)_{i=1}^k\) by
\((\ell_i)_{i\in I}\), while leaving \(D\) unchanged.  The two Gram
identities above restrict to indices in \(I\).  Hence its rows are pairwise
orthogonal; the pivot-block defects determine the coefficients of the
retained \(p_i\), and the nonsingularity of \(D\) determines those of the
\(S_s\).  Thus the restricted matrix has \(|I|+r\) independent rows in
\(K^{2(|I|+r)}\) and generates a self-dual code.  For
\(I=\varnothing\), this row space is
\[
 \{(z_1(1,c),\ldots,z_r(1,c)):z\in K^r\},
\]
the full terminal isotropic space in \(K^{2r}\), because \(D\) is
nonsingular.  This proves the final assertion.
\end{proof}

\begin{example}[A fixed rank-two pairing over \(\F_{5}\)]
\label{ex:rank-two-gf5}
\rm
Take \(c=2\), so \(c^2=-1\), and let \(\mathcal C\) be generated by
\[
G_0=\begin{pmatrix}
1&3&3&2&1&1&1&3\\
3&4&2&1&4&4&3&2\\
1&4&3&4&2&3&4&2\\
1&2&2&2&3&4&4&4
\end{pmatrix}.
\]
All calculations in this example are in \(\F_{5}\). One checks that
\(G_0G_0^{\mathsf T}=0\), while the minor formed by the first four columns
has determinant \(4\). Hence \(\mathcal C\) has dimension \(4\) and is self-dual.

Pair adjacent coordinates and let \(G_0^\sigma\) be obtained by ordering the
four two-coordinate blocks as \((2,4,1,3)\). Writing
\(a_i=Q_{ii}\), the diagonal formula in
\eqref{eq:universal-rank-r-box} gives \(a_1=4\) and \(a_2=3\). We
substitute these values and \(c=2\), but retain the displayed diagonal
formula in the second coordinate: the two blocks are
\((4,2\cdot4+1)\) and \((3,2\cdot3+1)\). The invertible row transformation
\[
T=\begin{pmatrix}
4&1&0&0\\
2&1&2&0\\
1&3&0&1\\
4&0&3&4
\end{pmatrix},\qquad \det T=3,
\]
gives
\[
G_{2,2}=TG_0^\sigma=
{\arrayrulewidth=1.2pt
\left[\begin{array}{c!{\vrule width 1.2pt}c!{\vrule width 0.6pt}cc}
(4,2\cdot4+1)&\multicolumn{1}{c!{\vrule width 0.6pt}}{2(1,2)}&(2,1)&(3,3)\\[-1pt]
\noalign{\hrule height 1.2pt}
4(1,2)&(3,2\cdot3+1)&(2,3)&(0,2)\\[-1pt]
\noalign{\hrule height 0.6pt}
(1,2)&4(1,2)&(1,2)&(1,2)\\
4(1,2)&2(1,2)&(1,2)&2(1,2)
\end{array}\right]}.
\]
This is \eqref{eq:universal-rank-r-box} with
\[
\ell_1=((2,1),(3,3)),\qquad
\ell_2=((2,3),(0,2)),\qquad
Q_{12}=2,\quad Q_{21}=4,
\]
and
\[
D=\begin{pmatrix}1&1\\1&2\end{pmatrix},\qquad \det D=1.
\]
Modulo five the diagonal pairs are \((4,4)\) and \((3,2)\), respectively.
For a block \((u,v)\), recall that \(\partial_2(u,v)=v-2u\).  Hence
\[
\begin{aligned}
\partial_2\ell_1
  &=\begin{pmatrix}1-2\cdot2\\3-2\cdot3\end{pmatrix}
    =\begin{pmatrix}2\\2\end{pmatrix},\\
-D\partial_2\ell_1
  &=-\begin{pmatrix}1&1\\1&2\end{pmatrix}
       \begin{pmatrix}2\\2\end{pmatrix}
    =-\begin{pmatrix}4\\1\end{pmatrix}
    =\begin{pmatrix}1\\4\end{pmatrix}.
\end{aligned}
\]
Similarly, \(\partial_2\ell_2=(4,2)^{\mathsf T}\) and
\(-D\partial_2\ell_2=(4,2)^{\mathsf T}\).
The off-diagonal relation
\eqref{eq:universal-rank-r-relation} is satisfied. Applying the blockwise
defect map gives
\[
\partial G_{2,2}=
\begin{pmatrix}
1&0&2&2\\
0&1&4&2\\
0&0&0&0\\
0&0&0&0
\end{pmatrix},
\]
so
\[
r=\dim(\mathcal C\cap U_2)=4-\rank(\partial G_{2,2})=2.
\]
The block permutation preserves \(U_2\); consequently, whole-block
permutations and row operations cannot reduce \(r\) for this fixed oriented
pairing.

Deleting pivot rows together with their corresponding block columns gives
\[
G_{1,2}=
{\arrayrulewidth=1.2pt
\left[\begin{array}{c!{\vrule width 1.2pt}cc}
(3,2\cdot3+1)&(2,3)&(0,2)\\[-1pt]
\noalign{\hrule height 1.2pt}
4(1,2)&(1,2)&(1,2)\\
2(1,2)&(1,2)&2(1,2)
\end{array}\right]},\qquad
G_{0,2}=
\begin{pmatrix}
(1,2)&(1,2)\\
(1,2)&2(1,2)
\end{pmatrix}.
\]
Thus Theorem~\ref{thm:qary-universal-rank-r} gives the exact self-dual chain
\[
[8,4]\longrightarrow[6,3]\longrightarrow[4,2],
\]
with the same terminal matrix \(D\) at every level and
\(\langle G_{0,2}\rangle_K=K(1,2)\oplus K(1,2)\).
\end{example}

\noindent
The formal proof separates fixed-pairing normalization from the forward
self-duality calculation:
\begin{lstlisting}[caption={Lean goals for Theorem~\ref{thm:qary-universal-rank-r}.}]
theorem every_qary_selfDualCode_has_rankBoxed_normalForm
    (c : K) (hc : c ^ 2 = (-1 : K))
    (h2 : (2 : K) != 0)
    (hC : QaryBlockSelfDualCode C) :
    HasQaryRankBoxedNormalForm c C

theorem paperRankBoxedRows_forward_selfDual
    (hc : c * c = -1) (h2 : (2 : K) != 0)
    (hD : RankBoxCoreFullRank D)
    (hoff : PaperOffDiagonalRelations c Q ell) :
    RankBoxedPairwiseOrthogonal
      (paperRankBoxedRows c Q ell D) /\
    LinearIndependent K (paperRankBoxedRows c Q ell D) /\
    paperSelfDualCode
      (rankBoxedRowSpace (paperRankBoxedRows c Q ell D))
\end{lstlisting}

\begin{remark}
\rm
The split boxed form of Theorem~\ref{thm:split-boxed-form} is the
case \(r=1\), \(D=(1)\), and \(\ell_i\cdot\ell_i=-1\) of
Theorem~\ref{thm:qary-universal-rank-r}, since then the forced diagonal
coefficient is \(Q_{ii}=\frac c2(1+\ell_i\cdot\ell_i)=0\).
\end{remark}

% Generated by check_gf13_repeated_lineage.py; do not edit.
\newcommand{\GFThirteenRepeatedLineageRows}{%
0 & $C_{18}^{(13)}$ & $[18,9,8]$ & 1,896 & rank-one parent \\
1 & $C_{20}^{(13)}$ & $[20,10,10]$ & $\mathbf{111,024}$ & $\mathcal B_5(G_{18}^{(13)},x_{18}^{(13)})$ \\
}
\newcommand{\GFThirteenRepeatedCatalogueRows}{%
\href{https://github.com/LeGenAI/intersection-coding-theory-cohomology/blob/main/Formalization/Verification/Examples/certificates/gf13-repeated-lineage.json}{$C_{18}^{(13)}$} & 13 & $[18,9,8];\,1,896$ & $(8;\,1,752)$ & \cite{KimChoi2021} \\
\href{https://github.com/LeGenAI/intersection-coding-theory-cohomology/blob/main/Formalization/Verification/Examples/certificates/gf13-repeated-lineage.json}{$C_{20}^{(13)}$} & 13 & $[20,10,10];\,\mathbf{111,024}$ & $(10;\,111,024)$ & \cite{BetsumiyaEtAl} \\
}
\newcommand{\GFThirteenLargestRepeatedMatrix}{\left(\begin{array}{r|cc!{\vrule width 1.2pt}cc|cc|cc|cc|cc|cc|cc|cc!{\vrule width 1.2pt}cc}
\scriptstyle\text{rows} & \multicolumn{2}{c!{\vrule width 1.2pt}}{\scriptstyle\text{new pair}} & \multicolumn{16}{c!{\vrule width 1.2pt}}{\scriptstyle \operatorname{diag}:\ 01;\quad i\ne j:\ b_{ij}(1,c)} & \multicolumn{2}{c}{\scriptstyle D=(1):\ \ell_i;\ (1,c)} \\ \noalign{\hrule height 0.7pt}
\scriptstyle 01\mid x & \cellcolor{black!7}\textcolor{blue!65!black}{0} & \cellcolor{black!7}\textcolor{blue!65!black}{1} & \textcolor{blue!65!black}{2} & \textcolor{blue!65!black}{0} & \textcolor{blue!65!black}{0} & \textcolor{blue!65!black}{11} & \textcolor{blue!65!black}{0} & \textcolor{blue!65!black}{5} & \textcolor{blue!65!black}{0} & \textcolor{blue!65!black}{0} & \textcolor{blue!65!black}{6} & \textcolor{blue!65!black}{7} & \textcolor{blue!65!black}{6} & \textcolor{blue!65!black}{6} & \textcolor{blue!65!black}{0} & \textcolor{blue!65!black}{9} & \textcolor{blue!65!black}{0} & \textcolor{blue!65!black}{0} & \textcolor{blue!65!black}{0} & \textcolor{blue!65!black}{1} \\ \noalign{\hrule height 1.2pt}
\scriptstyle b_{ij},\ell_i & \textcolor{orange!80!black}{2} & \textcolor{orange!80!black}{3} & \cellcolor{black!7}\textcolor{teal!60!black}{0} & \cellcolor{black!7}\textcolor{teal!60!black}{1} & \textcolor{teal!60!black}{10} & \textcolor{teal!60!black}{11} & \textcolor{teal!60!black}{12} & \textcolor{teal!60!black}{8} & \textcolor{teal!60!black}{7} & \textcolor{teal!60!black}{9} & \textcolor{teal!60!black}{9} & \textcolor{teal!60!black}{6} & \textcolor{teal!60!black}{0} & \textcolor{teal!60!black}{0} & \textcolor{teal!60!black}{6} & \textcolor{teal!60!black}{4} & \textcolor{teal!60!black}{7} & \textcolor{teal!60!black}{9} & \textcolor{teal!60!black}{9} & \textcolor{teal!60!black}{3} \\
\scriptstyle b_{ij},\ell_i & \textcolor{orange!80!black}{0} & \textcolor{orange!80!black}{0} & \textcolor{teal!60!black}{11} & \textcolor{teal!60!black}{3} & \cellcolor{black!7}\textcolor{teal!60!black}{0} & \cellcolor{black!7}\textcolor{teal!60!black}{1} & \textcolor{teal!60!black}{9} & \textcolor{teal!60!black}{6} & \textcolor{teal!60!black}{4} & \textcolor{teal!60!black}{7} & \textcolor{teal!60!black}{2} & \textcolor{teal!60!black}{10} & \textcolor{teal!60!black}{2} & \textcolor{teal!60!black}{10} & \textcolor{teal!60!black}{11} & \textcolor{teal!60!black}{3} & \textcolor{teal!60!black}{0} & \textcolor{teal!60!black}{0} & \textcolor{teal!60!black}{9} & \textcolor{teal!60!black}{3} \\
\scriptstyle b_{ij},\ell_i & \textcolor{orange!80!black}{2} & \textcolor{orange!80!black}{3} & \textcolor{teal!60!black}{5} & \textcolor{teal!60!black}{12} & \textcolor{teal!60!black}{8} & \textcolor{teal!60!black}{1} & \cellcolor{black!7}\textcolor{teal!60!black}{0} & \cellcolor{black!7}\textcolor{teal!60!black}{1} & \textcolor{teal!60!black}{9} & \textcolor{teal!60!black}{6} & \textcolor{teal!60!black}{10} & \textcolor{teal!60!black}{11} & \textcolor{teal!60!black}{6} & \textcolor{teal!60!black}{4} & \textcolor{teal!60!black}{10} & \textcolor{teal!60!black}{11} & \textcolor{teal!60!black}{12} & \textcolor{teal!60!black}{8} & \textcolor{teal!60!black}{5} & \textcolor{teal!60!black}{0} \\
\scriptstyle b_{ij},\ell_i & \textcolor{orange!80!black}{6} & \textcolor{orange!80!black}{9} & \textcolor{teal!60!black}{9} & \textcolor{teal!60!black}{6} & \textcolor{teal!60!black}{12} & \textcolor{teal!60!black}{8} & \textcolor{teal!60!black}{7} & \textcolor{teal!60!black}{9} & \cellcolor{black!7}\textcolor{teal!60!black}{0} & \cellcolor{black!7}\textcolor{teal!60!black}{1} & \textcolor{teal!60!black}{6} & \textcolor{teal!60!black}{4} & \textcolor{teal!60!black}{4} & \textcolor{teal!60!black}{7} & \textcolor{teal!60!black}{3} & \textcolor{teal!60!black}{2} & \textcolor{teal!60!black}{0} & \textcolor{teal!60!black}{0} & \textcolor{teal!60!black}{10} & \textcolor{teal!60!black}{4} \\
\scriptstyle b_{ij},\ell_i & \textcolor{orange!80!black}{7} & \textcolor{orange!80!black}{4} & \textcolor{teal!60!black}{1} & \textcolor{teal!60!black}{5} & \textcolor{teal!60!black}{8} & \textcolor{teal!60!black}{1} & \textcolor{teal!60!black}{0} & \textcolor{teal!60!black}{0} & \textcolor{teal!60!black}{12} & \textcolor{teal!60!black}{8} & \cellcolor{black!7}\textcolor{teal!60!black}{0} & \cellcolor{black!7}\textcolor{teal!60!black}{1} & \textcolor{teal!60!black}{1} & \textcolor{teal!60!black}{5} & \textcolor{teal!60!black}{8} & \textcolor{teal!60!black}{1} & \textcolor{teal!60!black}{11} & \textcolor{teal!60!black}{3} & \textcolor{teal!60!black}{3} & \textcolor{teal!60!black}{9} \\
\scriptstyle b_{ij},\ell_i & \textcolor{orange!80!black}{2} & \textcolor{orange!80!black}{3} & \textcolor{teal!60!black}{10} & \textcolor{teal!60!black}{11} & \textcolor{teal!60!black}{8} & \textcolor{teal!60!black}{1} & \textcolor{teal!60!black}{7} & \textcolor{teal!60!black}{9} & \textcolor{teal!60!black}{5} & \textcolor{teal!60!black}{12} & \textcolor{teal!60!black}{3} & \textcolor{teal!60!black}{2} & \cellcolor{black!7}\textcolor{teal!60!black}{0} & \cellcolor{black!7}\textcolor{teal!60!black}{1} & \textcolor{teal!60!black}{10} & \textcolor{teal!60!black}{11} & \textcolor{teal!60!black}{9} & \textcolor{teal!60!black}{6} & \textcolor{teal!60!black}{0} & \textcolor{teal!60!black}{5} \\
\scriptstyle b_{ij},\ell_i & \textcolor{orange!80!black}{12} & \textcolor{orange!80!black}{5} & \textcolor{teal!60!black}{3} & \textcolor{teal!60!black}{2} & \textcolor{teal!60!black}{11} & \textcolor{teal!60!black}{3} & \textcolor{teal!60!black}{8} & \textcolor{teal!60!black}{1} & \textcolor{teal!60!black}{7} & \textcolor{teal!60!black}{9} & \textcolor{teal!60!black}{8} & \textcolor{teal!60!black}{1} & \textcolor{teal!60!black}{3} & \textcolor{teal!60!black}{2} & \cellcolor{black!7}\textcolor{teal!60!black}{0} & \cellcolor{black!7}\textcolor{teal!60!black}{1} & \textcolor{teal!60!black}{6} & \textcolor{teal!60!black}{4} & \textcolor{teal!60!black}{8} & \textcolor{teal!60!black}{0} \\
\scriptstyle b_{ij},\ell_i & \textcolor{orange!80!black}{7} & \textcolor{orange!80!black}{4} & \textcolor{teal!60!black}{2} & \textcolor{teal!60!black}{10} & \textcolor{teal!60!black}{9} & \textcolor{teal!60!black}{6} & \textcolor{teal!60!black}{6} & \textcolor{teal!60!black}{4} & \textcolor{teal!60!black}{10} & \textcolor{teal!60!black}{11} & \textcolor{teal!60!black}{5} & \textcolor{teal!60!black}{12} & \textcolor{teal!60!black}{4} & \textcolor{teal!60!black}{7} & \textcolor{teal!60!black}{2} & \textcolor{teal!60!black}{10} & \cellcolor{black!7}\textcolor{teal!60!black}{0} & \cellcolor{black!7}\textcolor{teal!60!black}{1} & \textcolor{teal!60!black}{8} & \textcolor{teal!60!black}{0} \\ \noalign{\hrule height 1.2pt}
\scriptstyle a_i,(1,c) & \textcolor{orange!80!black}{3} & \textcolor{orange!80!black}{11} & \textcolor{violet!75!black}{3} & \textcolor{violet!75!black}{2} & \textcolor{violet!75!black}{3} & \textcolor{violet!75!black}{2} & \textcolor{violet!75!black}{12} & \textcolor{violet!75!black}{8} & \textcolor{violet!75!black}{7} & \textcolor{violet!75!black}{9} & \textcolor{violet!75!black}{6} & \textcolor{violet!75!black}{4} & \textcolor{violet!75!black}{8} & \textcolor{violet!75!black}{1} & \textcolor{violet!75!black}{1} & \textcolor{violet!75!black}{5} & \textcolor{violet!75!black}{1} & \textcolor{violet!75!black}{5} & \cellcolor{black!7}\textcolor{violet!75!black}{1} & \cellcolor{black!7}\textcolor{violet!75!black}{5} \\
\end{array}\right)}
\newcommand{\GFThirteenLargestParentParameters}{c=5,\; k=8}
\newcommand{\GFThirteenEighteenCorrection}{(2,0,0,11,0,5,0,0,6,7,6,6,0,9,0,0,0,1)}
\newcommand{\GFThirteenRepeatedCertificate}{\href{https://github.com/LeGenAI/intersection-coding-theory-cohomology/blob/main/Formalization/Verification/Examples/certificates/gf13-repeated-lineage.json}{$C_{20}^{(13)}$}}

\section{Applications}\label{sec:applications}

We first retain the explicit small-length applications from the submitted
manuscript: two optimal self-dual codes over \(\F_{5}\), two optimal
self-dual codes over \(\F_{13}\), and a self-dual \([12,6,6]\) code over
\(\F_{13}\).  We then add one repeated \(\F_{13}\) realization in which the
split-boxed parent and its building-up child are visible in the same complete
generator matrix.

\subsection{Optimal self-dual codes over \texorpdfstring{\(\F_{5}\)}{F5}}

\begin{prop}[An optimal self-dual \texorpdfstring{$[6,3,4]$}{[6,3,4]} code over \texorpdfstring{$\F_{5}$}{F5}]
\label{ex:split-boxed-63}
\rm
Let \(K=\F_{5}\) and \(c=2\), so \(c^2=-1\). Choose
\[
\ell_1=(0,2),\qquad \ell_2=(0,2),\qquad b_{12}=1,\qquad b_{21}=2.
\]
Then \eqref{eq:split-boxed-last-row} gives
\[
a_1=a_2=3,
\]
and the corresponding split boxed matrix is
\[
\widetilde M_3=
\begin{pmatrix}
01 & 12 & 02\\
24 & 01 & 02\\
31 & 31 & 12
\end{pmatrix}.
\]
This matrix generates an optimal self-dual \([6,3,4]\) code over \(\F_{5}\).
\end{prop}

\begin{proof}
The boxed relations are satisfied, so Theorem~\ref{thm:split-boxed-form} shows that \(\widetilde M_3\) generates a self-dual \([6,3]\) code. In ordinary coordinates the generator matrix is
\[
G_3=
\left(
\begin{array}{cc|cc|cc}
0&1&1&2&0&2\\
\hline
2&4&0&1&0&2\\
\hline
3&1&3&1&1&2
\end{array}
\right).
\]
Enumeration of the \(5^3\) coefficient vectors gives the weight enumerator
\[
1+60z^4+24z^5+40z^6.
\]
Hence the minimum distance is \(4\).  The Singleton bound gives
\(d\le6-3+1=4\) for every linear \([6,3]\) code, so this code is optimal.

Deleting the first block row and the first block column from \(\widetilde M_3\) leaves the shorter parent package
\[
\widetilde M_2=
\begin{pmatrix}
01 & 02\\
31 & 12
\end{pmatrix},
\]
which is the corresponding self-dual \([4,2,2]\) parent.
\end{proof}

\begin{prop}[A split boxed child \texorpdfstring{$[8,4,4]$}{[8,4,4]} over \texorpdfstring{$\F_{5}$}{F5}]
\label{ex:split-boxed-84}
\rm
Let \(K=\F_{5}\) and \(c=2\). Choose
\[
\ell_1=\ell_2=\ell_3=(0,2),
\]
and
\[
b_{12}=b_{13}=b_{23}=1,\qquad
b_{21}=b_{31}=b_{32}=2.
\]
Then \eqref{eq:split-boxed-last-row} gives
\[
a_1=a_2=a_3=3,
\]
and the corresponding split boxed matrix is
\[
\widetilde M_4=
\begin{pmatrix}
01 & 12 & 12 & 02\\
24 & 01 & 12 & 02\\
24 & 24 & 01 & 02\\
31 & 31 & 31 & 12
\end{pmatrix}.
\]
This matrix generates an optimal self-dual \([8,4,4]\) code over \(\F_{5}\).
\end{prop}

\begin{proof}
Again the boxed relations hold, so Theorem~\ref{thm:split-boxed-form} gives a self-dual \([8,4]\) code. In ordinary coordinates the generator matrix is

\[
G_4= \left(
\begin{array}{cc|cccc|cc}
0&1&1&2&1&2&0&2\\
2&4&0&1&1&2&0&2\\
2&4&2&4&0&1&0&2\\
\hline
3&1&3&1&3&1&1&2
\end{array}
\right).
\]
Enumeration of the \(5^4\) coefficient vectors gives
\[
1+48z^4+32z^5+288z^6+128z^7+128z^8,
\]
so the minimum distance is \(4\).  If a linear \([8,4]\) code over
\(\F_{5}\) had distance at least \(5\), the Singleton bound would force
distance \(5\), making it a nontrivial MDS code of length \(8>q+1=6\),
which is impossible.  Thus the displayed self-dual code is optimal.
\end{proof}

Moreover, deleting the first block row and the first block column from
\(\widetilde M_4\) recovers exactly \(\widetilde M_3\).  Thus these two
examples constitute one recursive calculation,
\[
[4,2,2]_5\longrightarrow[6,3,4]_5
\longrightarrow[8,4,4]_5,
\]
rather than two unrelated existence examples.

\subsection{Small-length self-dual codes over
\texorpdfstring{\(\F_{13}\)}{F13}}

The split boxed construction over \(\F_{13}\) gives the following three
applications from the submitted manuscript.

\begin{prop}[An optimal self-dual \texorpdfstring{$[8,4,5]$}{[8,4,5]} code over \texorpdfstring{$\F_{13}$}{F13}]
\label{prop:gf13-845}
\rm
Let \(K=\F_{13}\) and \(c=5\), so \(c^2=-1\). Choose
\[
\ell_1=(0,5),\qquad \ell_2=(0,5),\qquad \ell_3=(4,3),
\]
and
\[
b_{12}=11,\quad b_{21}=10,\quad b_{13}=5,\quad b_{31}=5,
\quad b_{23}=11,\quad b_{32}=12.
\]
Then \eqref{eq:split-boxed-last-row} gives
\[
a_1=8,\qquad a_2=8,\qquad a_3=4,
\]
and the corresponding split boxed matrix is
\[
G_8=
\left(
\begin{array}{cc|cccc|cc}
0&1&11&3&5&12&0&5\\
10&11&0&1&11&3&0&5\\
5&12&12&8&0&1&4&3\\
\hline
8&1&8&1&4&7&1&5
\end{array}
\right).
\]
This matrix generates an optimal self-dual \([8,4,5]\) code over
\(\F_{13}\).
\end{prop}

\begin{proof}
The boxed relations are satisfied, so
Theorem~\ref{thm:split-boxed-form} shows that \(G_8\) generates a
self-dual \([8,4]\) code.  The coefficient vector \((0,1,3,1)\) produces
the codeword
\[
(7,9,5,0,2,0,0,6),
\]
which has weight \(5\).  Exact enumeration of the \(13^4\) coefficient
vectors gives no nonzero word of smaller weight.  Hence the minimum distance
is \(5\).  The Griesmer bound excludes a linear \([8,4,6]\) code, since
\[
\sum_{i=0}^{3}\left\lceil\frac{6}{13^i}\right\rceil
=6+1+1+1=9>8.
\]
Thus \(G_8\) is optimal among all linear \([8,4]\) codes over \(\F_{13}\).
\end{proof}

\begin{prop}[An optimal self-dual \texorpdfstring{$[10,5,6]$}{[10,5,6]} code over \texorpdfstring{$\F_{13}$}{F13}]
\label{prop:gf13-1056}
\rm
Let \(K=\F_{13}\) and \(c=5\). Choose
\[
\ell_1=(9,3),\qquad \ell_2=(10,4),\qquad
\ell_3=(10,4),\qquad \ell_4=(4,10),
\]
and
\[
\begin{aligned}
&b_{12}=5,\ b_{21}=11,\quad b_{13}=11,\ b_{31}=5,\quad
b_{14}=10,\ b_{41}=8,\\
&b_{23}=7,\ b_{32}=1,\quad b_{24}=11,\ b_{42}=12,\quad
b_{34}=1,\ b_{43}=9.
\end{aligned}
\]
Then \eqref{eq:split-boxed-last-row} gives
\[
a_1=3,\qquad a_2=7,\qquad a_3=7,\qquad a_4=10,
\]
and the corresponding split boxed matrix is
\[
G_{10}=
\left(
\begin{array}{cc|cccccc|cc}
0&1&5&12&11&3&10&11&9&3\\
11&3&0&1&7&9&11&3&10&4\\
5&12&1&5&0&1&1&5&10&4\\
8&1&12&8&9&6&0&1&4&10\\
\hline
3&2&7&9&7&9&10&11&1&5
\end{array}
\right).
\]
This matrix generates an optimal self-dual \([10,5,6]\) code over
\(\F_{13}\).
\end{prop}

\begin{proof}
The boxed relations hold, so Theorem~\ref{thm:split-boxed-form} gives a
self-dual \([10,5]\) code.  The coefficient vector \((0,0,1,1,0)\)
produces
\[
(0,0,0,0,9,7,1,6,1,1),
\]
which has weight \(6\).  Exact enumeration of the \(13^5\) coefficient
vectors gives no nonzero word of smaller weight.  Thus the minimum distance
is \(6\).  The Griesmer bound excludes a linear \([10,5,7]\) code, since
\[
\sum_{i=0}^{4}\left\lceil\frac{7}{13^i}\right\rceil
=7+1+1+1+1=11>10.
\]
Therefore \(G_{10}\) is optimal among all linear \([10,5]\) codes over
\(\F_{13}\).
\end{proof}

\begin{prop}[An exact self-dual \texorpdfstring{$[12,6,6]$}{[12,6,6]} code over \texorpdfstring{$\F_{13}$}{F13}]
\label{prop:gf13-1266}
\rm
Let \(K=\F_{13}\) and \(c=5\). Choose
\[
\ell_1=(9,3),\quad \ell_2=(4,10),\quad \ell_3=(3,9),\quad
\ell_4=(10,4),\quad \ell_5=(9,3),
\]
and
\[
\begin{aligned}
&b_{12}=4,\ b_{21}=1,\quad b_{13}=7,\ b_{31}=3,\quad
b_{14}=7,\ b_{41}=9,\quad b_{15}=1,\ b_{51}=7,\\
&b_{23}=8,\ b_{32}=8,\quad b_{24}=0,\ b_{42}=10,\quad
b_{25}=8,\ b_{52}=10,\\
&b_{34}=12,\ b_{43}=6,\quad b_{35}=12,\ b_{53}=11,\quad
b_{45}=3,\ b_{54}=0.
\end{aligned}
\]
Then \eqref{eq:split-boxed-last-row} gives
\[
a_1=3,\qquad a_2=10,\qquad a_3=6,\qquad a_4=7,\qquad a_5=3,
\]
and the corresponding split boxed matrix is
\[
G_{12}=
\left(
\begin{array}{cc|cccccccc|cc}
0&1&4&7&7&9&7&9&1&5&9&3\\
1&5&0&1&8&1&0&0&8&1&4&10\\
3&2&8&1&0&1&12&8&12&8&3&9\\
9&6&10&11&6&4&0&1&3&2&10&4\\
7&9&10&11&11&3&0&0&0&1&9&3\\
\hline
3&2&10&11&6&4&7&9&3&2&1&5
\end{array}
\right).
\]
This matrix generates a self-dual \([12,6,6]\) code over \(\F_{13}\).
\end{prop}

\begin{proof}
The boxed relations give self-duality by
Theorem~\ref{thm:split-boxed-form}.  The coefficient vector
\[
(0,0,0,1,0,12)
\]
produces
\[
(6,4,0,0,0,0,6,5,0,0,9,12),
\]
which has weight \(6\).  Exact enumeration of the \(13^6\) coefficient
vectors gives no nonzero word of smaller weight.  Hence the minimum distance
is \(6\).
\end{proof}

\subsection{A repeated boxed realization over
\texorpdfstring{\(\F_{13}\)}{F13}}

Let \(c=5\in\F_{13}\), so \(c^2=-1\).  If \(G\) has rows
\(g_1,\ldots,g_m\), put
\begin{equation}
\mathcal B_5(G,x)=
\left[
\begin{array}{cc!{\vrule width 1.2pt}c}
0&1&x\\
\noalign{\hrule height 1.2pt}
-5y_1&-y_1&g_1\\
\vdots&\vdots&\vdots\\
-5y_m&-y_m&g_m
\end{array}
\right],
\qquad y_i=x\cdot g_i,\qquad x\cdot x=-1.
\label{eq:gf13-repeated-box}
\end{equation}
This is the matrix in Theorem~\ref{thm:qary-building-up} after interchanging
its first two coordinates.  Thus the new diagonal block is \((0,1)\), in the
same orientation as the diagonal blocks of the split-boxed parent.

\begin{prop}[A repeated boxed realization over
\texorpdfstring{\(\F_{13}\)}{F13}]
\label{ex:gf13-repeated-realization}
\rm
Let \(\mathcal C_{20}^{(13)}\) be the pure double-circulant self-dual code generated
by
\[
G_{20}=\left[I_{10}\ \middle|\
\operatorname{circ}(3,7,5,5,2,8,5,6,3,0)\right].
\]
It has parameters \([20,10,10]\) and \(A_{10}=111024\)
\cite{BetsumiyaEtAl}. After the ordered two-coordinate reduction with
column \(11\) first and column \(7\) second, denoted by \((11,7)\), followed
by the recorded row normalization and coordinate
permutation, one obtains a self-dual parent \(\mathcal C_{18}^{(13)}\) with parameters
\([18,9,8]\) and \(A_8=1896\).  There is a split-boxed generator
\(G_{18}^{(13)}\) of this parent and
\[
x_{18}^{(13)}=\GFThirteenEighteenCorrection
\]
such that
\[
x_{18}^{(13)}\cdot x_{18}^{(13)}=-1,\qquad
\left\langle\mathcal B_5
 \bigl(G_{18}^{(13)},x_{18}^{(13)}\bigr)\right\rangle_{\F_{13}}
 =\mathcal C_{20}^{(13)}.
\]
In \(G_{18}^{(13)}\), every pivot diagonal block is \((0,1)\), every
off-diagonal block is \(b_{ij}(1,5)\), and the terminal block is \((1,5)\).
Hence the parent is literally in the zero-diagonal split boxed form of
Theorem~\ref{thm:split-boxed-form}, while
\(\mathcal B_5(G_{18}^{(13)},x_{18}^{(13)})\) is its Kim--Lee
building-up child.
\end{prop}

\begin{proof}
Exact calculation over \(\F_{13}\) gives
\[
\rank G_{20}=10,\qquad G_{20}G_{20}^{\mathsf T}=0,
\]
and its complete weight distribution has first nonzero coefficient
\(A_{10}=111024\). The two-coordinate reduction at this ordered pair has rank
\(9\), zero Gram matrix, and first nonzero weight coefficient
\(A_8=1896\); it therefore generates the asserted self-dual parent.
The order is essential because the normalized top row has entries \((1,0)\)
in the first and second selected columns, respectively.

For the normalized parent, direct blockwise calculation verifies the
diagonal, off-diagonal, and terminal relations stated above.  Theorem~
\ref{thm:split-boxed-form} therefore gives its self-duality.  Moreover,
the displayed correction has norm \(-1\), and multiplication gives
\(y_i=x_{18}^{(13)}\cdot g_i\) for every parent row.  Swapping the first two
columns of \eqref{eq:gf13-repeated-box} therefore gives exactly the forward
matrix in Theorem~\ref{thm:qary-building-up}.  Its row space agrees with the
row space of \(G_{20}\) after the recorded permutation, which proves the
reconstruction identity.
\end{proof}

\begin{samepage}
Figure~\ref{fig:gf13-repeated-realization} displays the complete matrix, not
only its block pattern.  Blue marks the newly adjoined row, and orange marks
the new coordinate pair in the lower rows.  Inside the parent, teal marks
the pivot blocks, violet marks the terminal row and terminal pair, and light
gray marks the diagonal two-coordinate blocks.  The outer thick rules
separate the building-up border; the inner thick rules separate the pivot
part of the split-boxed parent from its terminal part.

\begin{center}
\centering
\resizebox{\textwidth}{!}{$
\widetilde G_{20}^{(13)}=\GFThirteenLargestRepeatedMatrix.
$}
\captionof{figure}{Repeated boxed realization
\(\mathcal C_{18}^{(13)}\longleftrightarrow\mathcal C_{20}^{(13)}\).}
\label{fig:gf13-repeated-realization}
\end{center}
\end{samepage}

Thus the same figure exhibits both directions:
\[
\mathcal C_{18}^{(13)}
\ \xrightarrow{\ \mathcal B_5(G_{18}^{(13)},x_{18}^{(13)})\ }\
\mathcal C_{20}^{(13)},
\qquad
\mathcal C_{20}^{(13)}
\ \xrightarrow{\ \substack{\text{columns }11,7\\\text{in this order}}\ }\
\mathcal C_{18}^{(13)}.
\]
Both distances equal the published best-known distances at their respective
lengths~\cite{KimChoi2021,BetsumiyaEtAl}.  The complete audit of all
\(20\cdot19=380\) ordered coordinate
pairs also identifies the limitation of this fixed child: every reduction
loses a weight-ten word, and the smallest loss is \(1896\), attained by ten
ordered pairs.  Hence this code does not produce a self-dual
\([18,9,9]\) parent by one two-coordinate reduction; this is a statement
about the fixed code, not a nonexistence theorem.

Thus one bidirectional calculation simultaneously displays the split-boxed
parent, the repeated building-up border, exact reconstruction, and
the behavior of the minimum-weight words.
%% ===========================================================================
\section{Conclusion}
\label{sec:conclusion}
%% ===========================================================================

We have separated the binary arithmetic comparison from the split
odd-characteristic geometry, proved the adapted \(q\)-ary forward and reverse
statements, and then removed the adapted-presentation hypothesis by a
universal rank-\(r\) normal form for every fixed pairing.  Thus the same
two-coordinate mechanism supports binary reduction, Kim--Lee extension, and
fixed-pairing normalization without identifying the binary Euclidean plane
with an alternating hyperbolic plane.  The applications retain the explicit
small-length codes over \(\F_{5}\) and \(\F_{13}\) from the submitted
manuscript and add the complete repeated boxed realization
\([18,9,8]_{13}\leftrightarrow[20,10,10]_{13}\), where the split-boxed parent
and its building-up child are visible in one generator matrix.  Lean verifies
the algebraic interfaces and their dependencies; the arithmetic realization
theorem and finite distance enumerations remain explicitly identified
mathematical and computational inputs.

%% ===========================================================================
\section*{Availability of the formalization}
%% ===========================================================================

The modular development uses Lean~\texttt{v4.29.0-rc6} and Mathlib.  Its
seventeen independent Comparator suites check 129 declarations, including
the exact fixed-pairing normal-form goal.  Production proofs use neither
\texttt{sorry} nor \texttt{native\_decide}; their transitive axioms are
confined to \texttt{propext}, \texttt{Quot.sound}, and
\texttt{Classical.choice}.  Reproducible finite-field calculations and the
application verification records are included in the same public repository:
\begin{center}
\href{https://github.com/LeGenAI/intersection-coding-theory-cohomology}
{\texttt{LeGenAI/intersection-coding-theory-cohomology}}.
\end{center}

\bibliographystyle{ieeetr}

\end{document}